\documentclass[12pt]{article}
\usepackage{url} 

\newcommand{\blind}{0}

\usepackage[utf8]{inputenc}
\usepackage{comment}
\usepackage{appendix}
\usepackage{natbib}
\usepackage{threeparttable}
\usepackage{varioref}
\usepackage[linktocpage=true]{hyperref}
\usepackage{amsmath}
\usepackage[capitalise,noabbrev]{cleveref}
\usepackage{amssymb}
\usepackage{amsthm}
\usepackage{graphicx}
\usepackage{float}
\usepackage{soul}
\usepackage{algorithm}
\usepackage[noend]{algpseudocode}
\usepackage{comment}
\usepackage{lmodern}

\newcommand{\vect}[1]{{\boldsymbol{#1}}} 
\DeclareMathOperator*{\argmin}{arg\,min} 
\def\one{\mbox{1\hspace{-4.25pt}\fontsize{12}{14.4}\selectfont\textrm{1}}} 

  {
      \theoremstyle{plain} 
      \newtheorem{assumption}{Assumption}
      \newtheorem{theorem}{Theorem}
      \newtheorem{lemma}{Lemma}
 }
\begin{document}

\def\spacingset#1{\renewcommand{\baselinestretch}%
{#1}\small\normalsize} \spacingset{1}


\if0\blind
{
  \title{Flexible Conformal Highest Predictive Conditional Density Sets}
  \author{Max Sampson \\
    Department of Statistics and Actuarial Science, University of Iowa\\
    and \\
    Kung-Sik Chan \\
    Department of Statistics and Actuarial Science, University of Iowa}
  \maketitle
} \fi

\bigskip
\begin{abstract}
We  introduce our method, conformal highest conditional density sets (CHCDS), that forms conformal prediction sets using existing estimated conditional highest density predictive regions. We prove the validity of the method, and that conformal adjustment is negligible under some regularity conditions. In particular, if we correctly specify the underlying conditional density estimator, the conformal adjustment will be negligible. 
The conformal adjustment, however, always provides guaranteed nominal unconditional coverage, even when the underlying model is incorrectly specified. 
We compare the proposed method via simulation and a real data analysis to other existing methods. Our numerical results show that CHCDS is better than existing methods in scenarios where the error term is multi-modal, and just as good as existing methods when the error terms are unimodal.
\end{abstract}

\noindent%
\vfill

\newpage
\spacingset{1.75} 

\section{Background}

The problem of estimating upper-level or highest density sets has been widely studied \citep{Polonik_1995, Cuevas1997_support_estimation, Rigollet_2009_density_set_estimation, chen2016density, Samworth_2010, Lei_Wasserman_Conformal_kernel}. It involves estimating $\{x: f(x) > \lambda^{(\alpha)} \}$ for some $\lambda^{(\alpha)} > 0$ and density $f$ based only on samples drawn from $f$, where $\lambda^{(\alpha)}$ is chosen such that $\int_{\{x:f(x) \leq \lambda^{(\alpha)} \}} f(y) dy = \alpha$. The unconditional density can also be replaced by a conditional density to find conditional density-level sets. Without further adjustment, one could use these sets to form prediction sets (herein, unadjusted highest density sets). Unadjusted highest density sets can be combined with conformal prediction to form adjusted highest density sets that guarantee finite sample coverage. Nearly all existing conformal methods that attempt to find highest density sets partition the data \citep{Lei_Wasserman_Conformal_kernel, izbicki2021cdsplit}, so the observed coverage probability can be highly variable within each partition \citep{Lei_Wasserman_Conformal_kernel, izbicki2021cdsplit, angelopoulos_gentle_introduction}. HPD-split uses conformal prediction and a conditional density estimator to approximate the highest density set, but requires numerical integration for each data point in the calibration set \citep{izbicki2021cdsplit}. It is well known that estimating the highest density sets is computationally easier (in one and multiple dimensions) than numerically integrating densities \citep{hyndman_conditional_density_1996}. Other conformal methods attempt to find the shortest interval, but these fail to capture error terms with multi-modality while requiring multiple models to be fit \citep{CHR, dcp}. More background on conformal prediction can be found in~\cref{sec: conformal prediction} of the Supplementary Material. See also \citep{conformal_book, dis_free_pred_COPS, angelopoulos_gentle_introduction} for an introduction to conformal prediction. 

In this paper we introduce a new conformal method for estimating the highest predictive density set or region that requires only an estimate of the conditional density, conformalized highest conditional density sets (CHCDS). Our method provides an easy to understand conformal adjustment to unadjusted highest density sets that guarantees finite sample coverage without partitioning the data. Our conformal score can be easily combined with any conditional density estimator to form conformal prediction regions. We compare our method to HPD-split, Conformal Histogram Regression (CHR), optimal Distributional Conformal Prediction (DCP), and Conformalized Quantile Regression (CQR) \citep{izbicki2021cdsplit, CHR, dcp, romanocqr}. The main advantages of CHCDS compared to existing methods are that it does not partition the data, it is computationally efficient as it requires no numerical integration and only requires fitting one model, and it can capture multi-model error terms.

\section{Algorithm idea} \label{sec:CHCDS_split}

Throughout we let $\vect{X}$ represent the covariates, $Y$ the response,  $(\vect{X}_1, Y_1), \ldots, (\vect{X}_n, Y_n)$ the observed data, and $\vect{X}_{n+1}$ the covariates for our prediction of interest.

We now describe our method, conformalized highest conditional density sets (CHCDS), which guarantees coverage for existing estimated conditional highest density predictive regions without creating partitions. As with other split conformal prediction methods, we begin by splitting our data into sets used for training and calibration. The training set is indexed by $\mathcal{I}_{tr}$ and the calibration set by $\mathcal{I}_{cal}$. Given any conditional density estimating function, $\mathcal{B}$, we fit $\hat{f}(\cdot\mid \vect{x})$ on the training data:

\[
\hat{f} \leftarrow \mathcal{B}[\{(Y_i, \vect{X}_i): i \in \mathcal{I}_{tr}\}].
\]

Now, using the trained model, compute unadjusted $1 - \alpha$ highest density prediction sets using $\hat{f}(\cdot\mid\vect{X}_i)$ for each $i \in \mathcal{I}_{cal}$. Denote the density cutoff points as $\hat{c}(\vect{X}_i)$, whereas their true counterparts w.r.t. $f(\cdot\mid \vect{X}_i)$ are denoted as $c(\vect{X}_i)$. The unadjusted highest density set would then be,
\[
\{y: \hat{f}(y|\vect{X}_{n+1}) > \hat{c}(\vect{X}_{n+1}) \}.
\]

Instead of using the set with no coverage guarantees, we use the density cutoff points and the calibration set to compute the scores,
\begin{equation}
  V_{i} = \hat{f}(Y_i\mid \vect{X}_i) - \hat{c}(\vect{X}_i) \text{, } \forall i \in \mathcal{I}_{cal}. \label{eq: defn of V}  
\end{equation}

Next, compute $\hat{q} =\lfloor \alpha (n_{cal} + 1) \rfloor$th smallest value of $\{ \vect{V} \}$. For the new data point, $\vect{X}_{n+1}$, compute the unadjusted $1 - \alpha$ highest density prediction set using $\hat{f}(\cdot\mid \vect{X}_{n+1})$. Denote the estimated cutoff point as $\hat{c}(\vect{X}_{n+1})$. The conformal prediction set then becomes
\[
\vect{C}(\vect{X}_{n+1}) = \{y: \hat{f}(y\mid \vect{X}_{n+1}) > \hat{c}(\vect{X}_{n+1}) + \hat{q} \}.
\]

For reference, the procedure is outlined in Algorithm~\ref{alg:CHCDS_split}. Intuitively, this method starts with an estimate of the highest density set. It then adjusts the cutoff point based on the scores to achieve the desired coverage. This is similar in nature to conformalized quantile regression, but it adjusts vertically based on the density estimate instead of horizontally based on the true response \citep{romanocqr}. A visual of how our score works can be seen in Section~\ref{sec: conformal prediction} in the Supplementary Material. If there is substantive knowledge that the predictive distribution is unimodal, then existing unimodal density estimators may be used to ensure CHCDS outputs a prediction interval. See~\cref{sec:unimodal density estimation} in the Supplementary Material for further discussion on unimodal density estimation and a short discussion on the benefits of a multi-modal density estimator when there is a non-standard relationship, for example multifunctional covariate-response relationships.


\begin{algorithm}
    \caption{CHCDS}\label{alg:CHCDS_split}
    \textbf{Input:} level $\alpha$, data = $\mathcal{Z} = (Y_i, \vect{X}_i)_{i \in \mathcal{I}}$, test point $(\vect{x})$, and conditional density algorithm $\mathcal{B}$ \newline
    \textbf{Procedure:}
    \begin{algorithmic}[1]

    \State Split $\mathcal{Z}$ into a training fold $\mathcal{Z}_{tr} \overset{\Delta}{=} (Y_i, \vect{X}_i)_{i \in \mathcal{I}_{tr}}$ and a calibration fold $\mathcal{Z}_{cal} \overset{\Delta}{=} (Y_i, \vect{X}_i)_{i \in \mathcal{I}_{cal}}$
    \State Fit $\hat{f} = \mathcal{B}(\{ (\vect{X}_i, Y_i): i \in \mathcal{I}_{tr} \})$
    \State For each $i \in \mathcal{I}_{cal}$, use a root finding technique to estimate the shortest $1 - \alpha$ upper-level density set using $\hat{f}(Y|\vect{X}_i)$. Denote the density cutoff points as $\hat{c}(\vect{X}_i)$
    \State For each $i \in \mathcal{I}_{cal}$, compute the scores $V_i = \hat{f}(Y_i|\vect{X}_i) - \hat{c}(\vect{X}_i)$
    \State Compute $\hat{q} =\lfloor \alpha (n_{cal} + 1) \rfloor$th smallest value of $\{ \vect{V} \}$ 
    \State For the test point, use a root finding technique to estimate the shortest $1 - \alpha$ set for $\hat{f}(Y|\vect{x})$. Denote the density cutoff point as $\hat{c}(\vect{x})$
    \end{algorithmic}
    
    \textbf{Output:} $\hat{C}(\vect{x}) = \{y: \hat{f}(y|\vect{x}) >\hat{c}(\vect{x}) + \hat{q} \}$ 
\end{algorithm}

An immediate question is: what is the advantage of CHCDS over using the conditional density as the score, $V_i = \hat{f}(Y_i \mid \vect{X}_i)$? Using the conditional density on its own can be viewed as a special case of CHCDS with $\hat{c}(\vect{x}) = 0$, $\forall x$. When viewed this way, the advantage of CHCDS is clear. CHCDS starts with a density set that is already attempting to control conditional coverage. It then provides an adjustment to the density set cutoff  to guarantee finite sample coverage. On the other hand, the conditional density score gives the same density set cutoff for each set, regardless of the covariate value. A representative simulation is given in the Supplementary Materials~\cref{sec:density_vs_chcds} to demonstrate that CHCDS has better conditional coverage than using the conditional density as the score.

A rare problem can occur when the prediction level is high and the conditional density estimator overestimates the estimated density value cutoff for the unadjusted $1 - \alpha$ prediction set. That is, the estimated prediction set under covers before the conformal adjustment. When this occurs, the final density cutoff for the new prediction point is a negative value, or an infinite prediction set. The problem can be mitigated by using a multiplicative conformal adjustment in lieu of an additive conformal adjustment. We chose not to use the multiplicative conformal adjustment in our numerical studies because we did not encounter any instances of infinite prediction sets.
However, the two approaches generally yield similar prediction sets, see~\cref{sec:division_subtraction_comparison} in the Supplementary Material for a detailed comparison.

\section{Theoretical Properties}~\label{sec: theories}

In this section, we study two theoretical properties of the proposed method. The following result is a standard conformal result that establishes the coverage rate of the proposed method. All proofs are relegated to Section~\ref{sec: proofs} in the Supplementary Material. 
 
\begin{theorem}~\label{thm: conformal validity}
If $(Y_i, \vect{X}_i), i = 1, \ldots, n$ are exchangeable, then the prediction interval $\hat{C}(\vect{X}_{n+1})$ constructed by CHCDS satisfies
\[
pr\{Y_{n + 1} \in \hat{C}(\vect{X}_{n+1})\} \geq 1 - \alpha.
\]
\noindent If the $V_i$'s are almost surely distinct, then
\[
pr\{Y_{n+1} \in \hat{C}(\vect{X}_{n+1})\} \leq 1 - \alpha + (n_{cal} + 1)^{-1}.
\]
\end{theorem}

Below, we show that if the conditional density function can be well estimated, then the conformal adjustment is asymptotically negligible. Thus, a small conformal adjustment may be suggestive of a good conditional density fit  and asymptotic conditional coverage rate of $(1-\alpha)\times 100\%$ can be obtained. The latter may be established under some regularity conditions, but we shall not pursue this issue here. Thus, the proposed method preserves the asymptotic $1-\alpha$ conditional coverage rate when the underlying conditional density estimation is correctly specified, while guaranteeing $1-\alpha$ unconditional coverage, even with an incorrectly specified conditional density estimator. The following regularity conditions will be needed below:

\begin{assumption}~\label{assumption: C1}
The sample space of $(Y, \vect{X}^\intercal)^\intercal$ is a compact set of the form $ [a,b] \times \mathcal{X}$ where $a<b$ are two finite numbers, and $\mathcal{X}$ is a compact subset of the $d$-dimensional Euclidean space.     
\end{assumption}
\begin{assumption}~\label{assumption: C2}
   The conditional density function $f(y\mid\vect{x})$ is continuous in $(y, \vect{x})$ and positive, over the sample space.  
\end{assumption}
 \begin{assumption}~\label{assumption: C3}
The conditional density estimator $\hat{f}(\cdot\mid\vect{x})$ converges to $f(\cdot\mid\vect{x})$ uniformly over the sample space, at the rate of $b_{n_{train}}$, in probability,  where $b_{n_{train}}\to 0$ as the training sample size $n_{train}\to\infty$. That is, there exists a constant $K$ such that the sup-norm  
$\|f(\cdot\mid\vect{x}) - \hat{f}(\cdot\mid\vect{x})\|_\infty \le K\times b_{n_{train}}$ for all $\vect{x}\in \mathcal{X}$, except for an  event whose probability converges to 0, as $n_{train}\to\infty$. Similarly, the conditional cumulative distribution function 
$\hat{F}(\cdot\mid\vect{x})$ converges to $F(\cdot\mid\vect{x})$ uniformly at the rate of $b_{n_{train}}$, in probability.      
 \end{assumption}   
 \begin{assumption}~\label{assumption: C4}
   The population quantile function of $f(Y\mid\vect{X})-c(\vect{X})$ is continuous at $\alpha$.    
 \end{assumption}

We remark that assumptions~ \ref{assumption: C1} and ~\ref{assumption: C2}  are mild regularity conditions.  The proof of Theorem~\ref{thm: convergence rate}, given in the Supplementary Material, shows that for fixed $0<\alpha<1$, $c(\vect{x})$ is a continuous function of $\vect{x}$, hence assumption~\ref{assumption: C4} holds under very general conditions. 
Sufficient conditions for assumption~\ref{assumption: C3} have been studied in the literature of kernel density estimation, see Section~\ref{sec:validity of C3} in the Supplementary Material. 

\begin{theorem}\label{thm: convergence rate}
Suppose assumptions~\ref{assumption: C1}--\ref{assumption: C4} hold, and let $1>\alpha>0$ be fixed. Then $\hat{c}(\vect{x})$ converges to $c(\vect{x})$  at the rate of $b_{n_{train}}$ uniformly in $\vect{x}$, with probability approaching 1. Moreover, the conformal adjustment $\hat{q}$, defined below \eqref{eq: defn of V}, is of the order $O_p(b_{n_{train}}+n_{cal}^{-1/2})$, as both $n_{cal}\to\infty$ and $n_{train}\to\infty$. 
\end{theorem}

An analogous theorem for parametric density estimation and the corresponding proof as well as a discussion on how these results compare to other conformal prediction methods can be found in the Supplementary Material Section~\ref{sec: proofs_thm3}.

For an intuitive understanding of~\cref{thm: convergence rate}, assume the true conditional density model is being used. Denote the $1 - \alpha$ highest density cutoff for $\vect{x}$ as $c(\vect{x})$. Then, $P\{{f}(Y \mid \vect{x}) < c(\vect{x})\} = \alpha$. Using the CHCDS scores, $V_i = f(Y_i \mid \vect{X}_i) - c(\vect{X}_i)$, we then know that, $P(V < 0) = \alpha$. So, with a good conditional density model and a reasonably large sample size, we expect the conformal adjustment to be small for a sufficient sample size. 

\section{Numerical Studies}\label{sec:chcds_simulation}
In this section, we demonstrate the empirical performance of CHCDS compared to that of HPD-split, DCP, CQR, and CHR from \cite{izbicki2021cdsplit, dcp, romanocqr, CHR}, in one scenario. A second simulation scenario can be found in the Supplementary Material Section~\ref{sec:further-chcds_simulation}. In the below scenario, the mixture scenario, one predictor was generated, $X \sim \text{Unif}(-1.5, 1.5)$.

     \noindent Mixture: $Y|X, p \sim p\mathcal{N}\{f(X) - g(X), \sigma^2(X)\} + (1 - p) \mathcal{N}\{f(X) + g(X), \sigma^2(X)\}$, where $f(x) = (x - 1)^2(x + 1)$, $g(x) = 2 \one (x \geq -0.5) (x + 0.5)^{1/2}$, $\sigma^2(x) = 0.25 + |x|$, and $p \sim \text{bin}(0.5)$.

We compared the coverage, average size of the prediction set, and conditional coverage absolute deviation, defined as the absolute difference between the coverage rate, $1 - \alpha = 0.90$, and the observed coverage at each value of $X$. Simulation standard errors are given in parentheses if they are larger than 0.001. For the FlexCode density estimator, used with the unadjusted highest density set and HPD-split, we used a Fourier basis and the regression functions were estimated with Nadaraya-Watson Regression. For CHCDS, we used both a conditional kernel density estimator computed only on the 75 nearest neighbors of $\vect{X}_i$ (KNN Kernel), and a Gaussian mixture distribution with four components for the joint density and two components for the marginal covariate density (Gaussian Mix), see Section~\ref{sec: conditional density estimation} in the Supplementary Material for further details. For DCP, CQR, and CHR, we used quantile forest models. All simulations had a simulation size of 10,000. Each scenario had 1,000 training samples and 500 calibration samples. Results can be found in~\cref{tab:CDE_HPD_Bimodal}. Graphical comparisons of conditional coverage can be found in Figure~\ref{fig:conditional_cov_bimodal}.
Plots showing examples of the prediction regions  can be found in the Supplementary Material (\cref{fig:regions_unadj_asymm} --~\cref{fig:regions_CHR_mixture}.) \newline

\begin{table}
\begin{center}
\caption{Comparison of the methods in the mixture scenario}
    \begin{tabular}{lccc}
         Approach &  Coverage & Set Size & Conditional Absolute Deviation   \\
        Unadjusted (FlexCode) & 925 (3) & 5878 (2) & 58 (1)\\
        HPD-split (FlexCode) & 900 (3) & 5466 (3) & 63 (1)\\
        CHCDS (KNN) & 906 (3) & 5319 (2) & 8 \\
        CHCDS (Gaussian Mix) & 903 (3) & 5156 (2) & 28 (1)\\
        DCP & 901 (3) & 5173 (1) & 5  \\
        CQR & 901 (3) & 5834 (2) & 21 (1) \\
        CHR & 899 (3) & 5741 (2) & 14 (1) \\
    \end{tabular}
    \label{tab:CDE_HPD_Bimodal}
\end{center}
    \caption {Monte Carlo error given in parentheses only if it is greater than 0.001. All values have been multiplied by $10^3$.}

\end{table}

 \begin{figure}
     \centering
     \begin{tabular}{cc}
  \includegraphics[scale = 0.28]{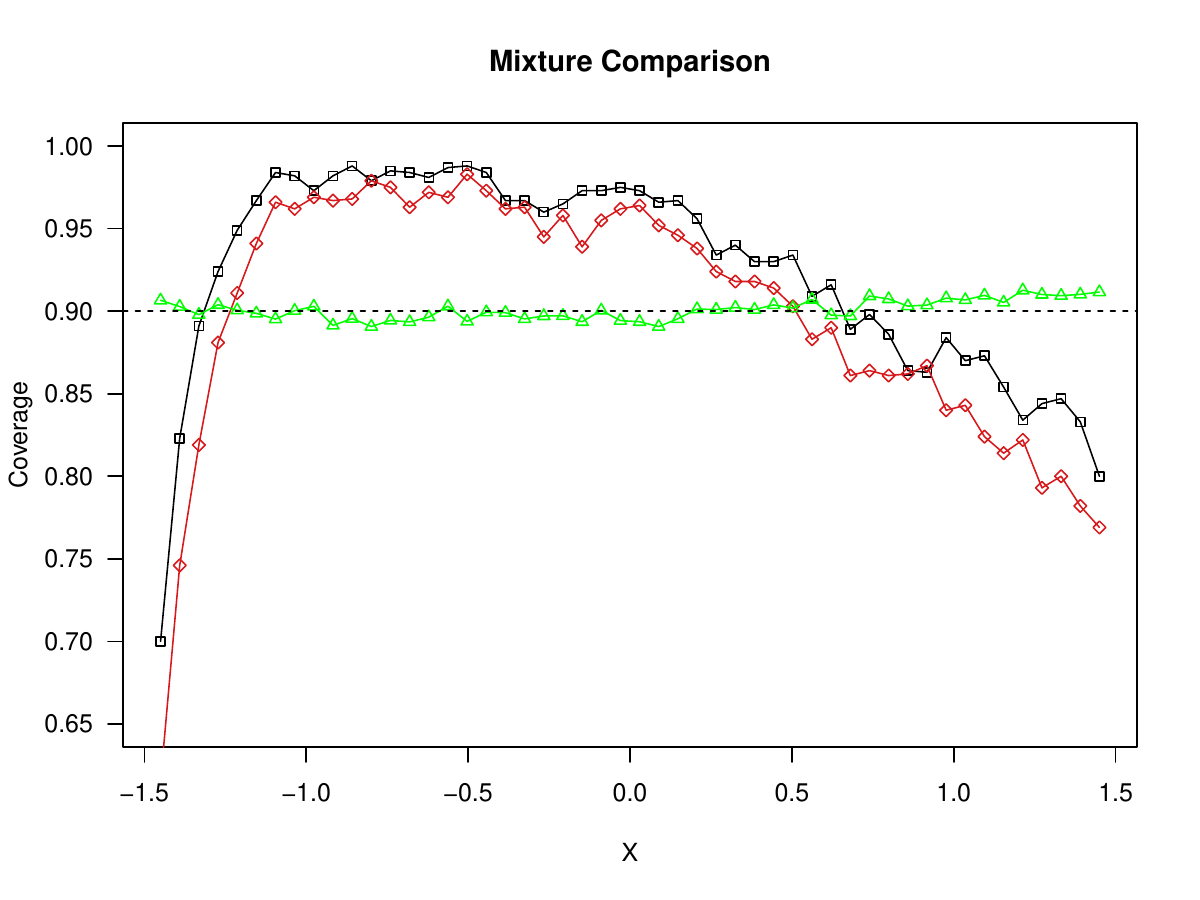}       &  \includegraphics[scale = 0.28]{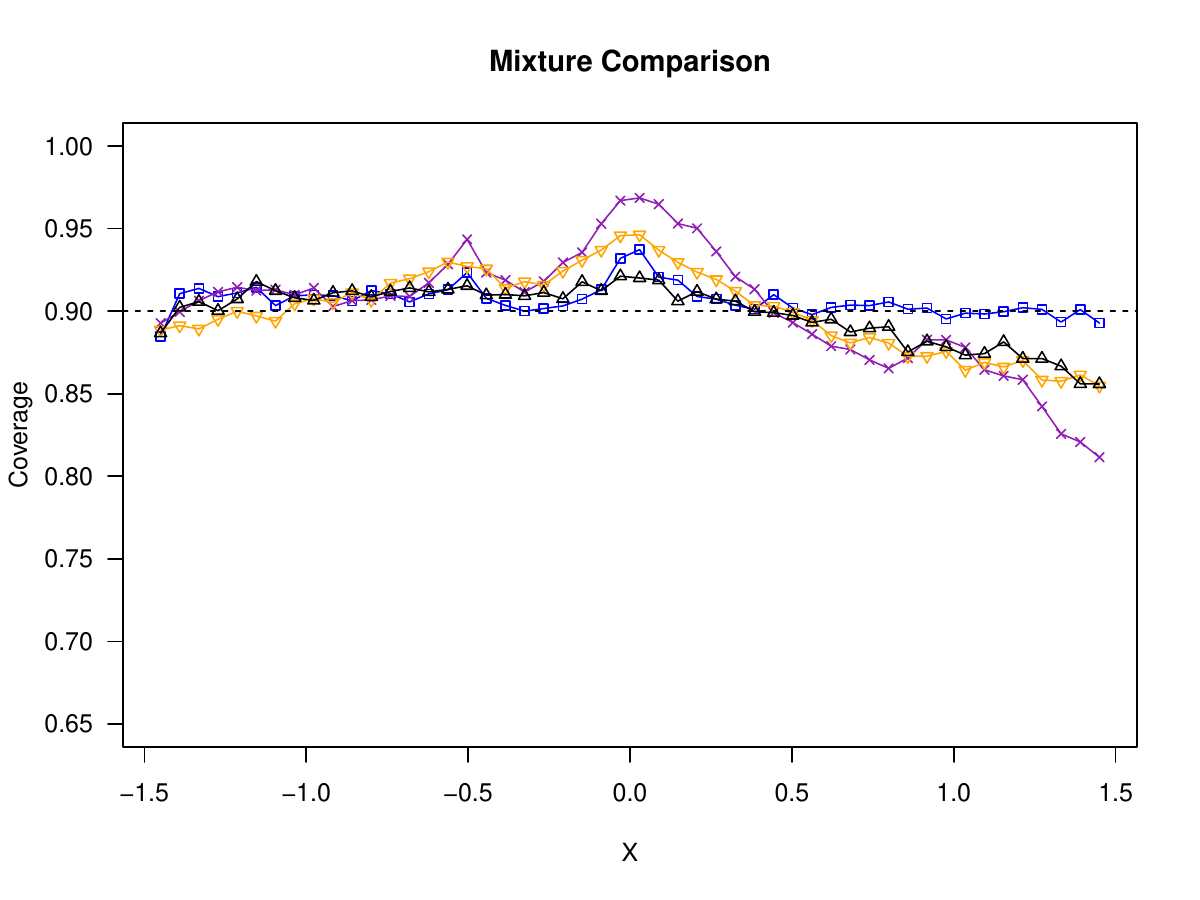}
    \end{tabular}

 \caption{A diagram showing the comparison of conditional coverage in the mixture scenario. Left: unadjusted FlexCode (black squares), HPD-split (red diamonds), DCP (green triangles). Right CHCDS (Gaussian Mix) (purple X), CHCDS (KNN) (blue circles), CQR (orange nablas), and CHR (black circles). The dashed line represents the desired 90\% coverage. The other lines represent the conditional coverage at a given value of $X$.}\label{fig:conditional_cov_bimodal}
 \end{figure}

The first note is that, clearly, the unadjusted approach has poor conditional coverage and is needlessly large. As expected, the results heavily depend on how good the initial model is. DCP and CHCDS with a KNN kernel density estimator performed very well with respect to conditional coverage. It is clear from the conditional coverage plots that all the prediction set approaches fall off in the tails except for CHCDS with the KNN kernel density estimator and optimal DCP. The KNN density estimator allows us to heavily weight nearby points, providing many of the benefits of forming a partition (conditional coverage near the nominal coverage rate), without creating variability in the observed coverage within each partition. See Section~\ref{sec:further-chcds_simulation} in the Supplementary Material for further discussion on the undercoverage seen in the tails. 

\section{Real Data Analysis}~\label{sec:real_dat_redshift}

For the Gaussian mixture density estimator, we used four mixture components for the joint density and three for the marginal covariate density on the HappyA dataset, which contains 74,950 galaxies \citep{Beck_Lin_happpya}. For all of the other models, we used the same setup as in~\cref{sec:chcds_simulation}. 
See Section~\ref{sec:further data analysis} in the Supplementary Material for an alternative implementation of the mixture model fitting, via neural network, which provided slightly better performance. 
Our goal was to predict galaxy redshift based on r-magnitude and the brightness of 4 non-red colors. It is known that the conditional density of redshift can be multi-modal \citep{Sheldon_2012_redshift, izbicki2021cdsplit, carrasco_2013_redshift}. We used 64,950 observations to train the model, 5,000 to compute non-conformity scores, and 5,000 to test out of sample prediction. This was repeated 10 times. We only compared unadjusted FlexCode, HPD-split, CHCDS-Gaussian Mix, DCP, CQR, and CHR because the KNN method was unable to handle the 7 dimensional data. For all of the FlexCode density estimators, we used 15,000 observations to train the model, 1,000 for calibration, and 1,000 for out of sample test predictions because using all of the data was not computationally feasible. Our goal was to create prediction sets with a coverage rate of 80\%, following \cite{izbicki2021cdsplit}. \cref{tab:redshift_conditional} reports the coverage, size, conditional coverage and size for bright and faint galaxies, and the time (s) to run one of the 10 iterations. We defined bright galaxies as those with an r-magnitude less than the median, 19.22, and faint as those with an r-magnitude greater than the median.

\begin{table}
\begin{center}
\caption{Conditional coverage and average size of the prediction regions for bright and faint galaxies.}

    \begin{tabular}{lccccccc}
         Type & Cov. & Size &B. Cov. & B. Size & F. Cov. & F. Size  & Time \\
         Unadjusted & 939 (1) & 332 (3) & 979 (2)  & 325 (5) & 900 (2) & 339 (3) & 66  \\
         HPD-split& 798 (3) & 189 (3)  & 865 (4)  & 183 (3) & 732 (6) & 187 (3)  & 120 \\
         CHCDS & 798 (2) & 107  & 799 (2) & 68 & 796 (3) & 146 (1) & 82 \\
         DCP & 803 (3) & 89 (1) & 816 (3) & 57  & 789 (4) & 121 (1) & 2608 \\
         CQR & 801 (2) & 96 (1) & 812 (2) & 60 & 791 (2) & 131 (1) & 1652 \\
         CHR& 798 (2) & 95  & 804 (2) & 59 (1) & 793 (2) & 131 (1) & 445 \\

    \end{tabular}
    \label{tab:redshift_conditional}
\end{center}
    \caption{
	Standard errors are given in parentheses if they are greater than 1. All values except time have been
multiplied by $10^3$. B. and F. represent bright and faint, respectively. Cov. stands for Coverage.}
\end{table}

We can see from the results that the FlexCode density estimators performed poorly. The approaches that used quantile forests (DCP, CQR, and CHR) had small set sizes, but failed to maintain approximate nominal coverage (80\%) with both bright and faint galaxies. The best conditional approach was CHCDS with a Gaussian mixture density estimator. It was able to keep a reasonably small overall set size, while having near nominal coverage for bright and faint galaxies; moreover, it was the fastest conformal approach.

\section{Discussion}~\label{sec: conclusion}

We introduced a new conformal score with  guaranteed coverage to existing highest density sets, CHCDS. The main benefits of CHCDS are that it is computationally efficient, only requires one model to be fit, captures multi-modal error terms, does not partition the data, and does not require numerical integration. Theoretically and numerically, we showed that when a good model is used, the conditional coverage is nearly perfect with prediction sets that are as small or smaller than other conformal methods. HPD-split can be used with the same conditional density estimators with similar results, but is more computationally intensive as it requires numerical integration to compute the non-conformity scores. 

A potential problem of CHCDS is  that the prediction sets are not necessarily convex. If there is substantive knowledge that the predictive distribution is unimodal, this can be solved by using a unimodal density estimator. However, when no such knowledge exists, the non-convex prediction sets can be indicative of a non-standard, for example, multifunctional, covariate-response relationship which can be further explored.

\section*{Acknowledgement}
This project was partially funded by National Institutes of Health Predoctoral Training Grant T32 HL 144461.

\newpage

\begin{center}
\textbf{\Large Supplementary material for Flexible Conformal Highest Predictive Conditional Density Sets}
\end{center}

\setcounter{equation}{0}
\setcounter{section}{0}
\setcounter{figure}{0}
\setcounter{table}{0}
\setcounter{page}{1}
\makeatletter
\renewcommand{\theequation}{S\arabic{equation}}
\renewcommand{\thefigure}{S\arabic{figure}}
\renewcommand{\thesection}{S\arabic{section}}
\renewcommand{\bibnumfmt}[1]{[S#1]}
\renewcommand{\citenumfont}[1]{S#1}

The Supplementary material comprises seven sections with the first section providing further background on conformal prediction and contrasting the conformity scores used by CHCDS, HPD-split, CQR, DCP, and CHR. The second section elaborates a multiplicative version of CHCDS and compares it with the additive version. The third section reviews conditional density estimation. The fourth section discusses some sufficient conditions for the validity of assumption~\ref{assumption: C3}. The fifth section contains all the proofs of the theorems stated in the main text as well as a comparison between the theoretical results of CHCDS and other conformal methods. The sixth section provides additional results for the simulation study. The last section elaborates an alternative implementation of the mixture model, via neural network.

\section{Conformal Prediction}~
\label{sec: conformal prediction}
Conformal prediction is a general method of creating prediction intervals that provide a non-asymptotic, distribution free coverage guarantee. Suppose we observe $n$ i.i.d. (more generally, exchangeable) copies of $\{(Y_1, \vect{X}_1), (Y_2, \vect{X}_2)\ldots, (Y_n, \vect{X}_n)\}$, with distribution $P$. Suppose that the first $n$ pairs are observed. Then, we want a set, $C(\vect{x}) = C_n((Y_1, \vect{X}_1),$ $(Y_2, \vect{X}_2),\ldots,$ $(Y_n, \vect{X}_n), \vect{x})$ such that for a new pair $(Y_{n+1}, \vect{X}_{n+1})$,
\begin{equation}
pr\{Y_{n+1} \in C(\vect{X}_{n+1})\} \geq 1-\alpha
\label{eq:p_dim_marginal_coverage}
\end{equation}
This coverage is guaranteed unconditionally \citep{conformal_book, dis_free_pred_COPS}. If we wanted to have finite sample, distribution-free, and conditional coverage for a continuous response,
\begin{equation}
pr\{Y_{n+1} \in C(\vect{X}_{n+1})\mid \vect{X}_{n+1}\} \geq 1-\alpha, \quad a.s.
\label{eq:conditional_coverage_conformal}
\end{equation}
our expected prediction set length would be infinite \citep{dis_free_pred_COPS}.

A method that attempts to approximate conditional coverage is locally valid conditional coverage. Let $\mathcal{A} = \{A_j: j\geq 1\}$ be a partition of the covariate space. A prediction set, $C(\vect{x})$, is locally valid with respect to $\mathcal{A}$ if 
\[
pr\{Y_{n+1} \in C(\vect{X}_{n+1})\mid \vect{X}_{n+1} \in A_j\} \geq 1 - \alpha, \text{  for all }  j.
\]
Local validity is achieved by computing conformal prediction sets using only the data within each partition. These sets tend to involve partitions of the data where prediction sets are formed based on density estimators within each partition \citep{izbicki2019flexible, izbicki2021cdsplit, Lei_Wasserman_Conformal_kernel}. When non-conformity score are almost surely distinct, the coverage probability when conditioning on both the training and calibration sets is a random quantity,
\[
\mathbb{P}(Y_{n+1} \in C(\vect{X}_{n+1}) \mid  (\vect{X}_i, Y_i), \> i = 1, \ldots ,n) \sim \text{Beta}(\kappa_{\alpha}, n_{cal} + 1 - \kappa_{\alpha}),
\] 
where $\kappa_{\alpha} = \lceil(1 - \alpha)(n_{cal} + 1)\rceil$ \citep{vovk_2012_conditional_validity, angelopoulos_gentle_introduction}. The idea is, if we run the conformal prediction algorithm multiple times, each time sampling a new finite observed dataset, then check the coverage on an infinite number of validation points, each coverage rate will be a draw  from the above Beta distribution. So, when using partitions  our prediction sets have more variability in coverage conditional on the observed data because we have a smaller sample size within each partition than we do overall \citep{Lei_Wasserman_Conformal_kernel, izbicki2021cdsplit, angelopoulos_gentle_introduction}. 

\subsection{Related Conformal Approaches}\label{sec:competing_methods}

Conformal quantile regression (CQR) attempts to control the conditional coverage by providing an additive conformal adjustment to an existing quantile regression prediction interval. Denote the unconformalized quantile regression interval trained on the training set as $(\hat{q}_{low}(\vect{x}), \hat{q}_{high}(\vect{x}))$. The non-conformity score is then, 
\[
V_i = \max\{\hat{q}_{low}(\vect{X}_i) - Y_i, Y_i - \hat{q}_{high}(\vect{X}_i)\}, \> i \in \mathcal{I}_{cal}.
\] 
The final conformal prediction interval is then,
\[
C(\vect{X}_{n+1}) = [\hat{q}_{low}(\vect{X}_{n+1}) - Q_{1-\alpha}(\vect{V}; \mathcal{Z}_{cal}), \hat{q}_{high}(\vect{X}_{n+1}) + Q_{1-\alpha}(\vect{V}; \mathcal{Z}_{cal})],
\]
where 
\[
Q_{1-\alpha}(\vect{V}; \mathcal{Z}_{cal}) := (1-\alpha)(1 + \frac{1}{|\mathcal{Z}_{cal}|})-\text{th empirical quantile of }\{V_i\}.
\]
and $|\mathcal{Z}_{cal}|$ is the size of the calibration set. While CQR attempts to control conditional coverage, it does not try to find the smallest intervals \citep{romanocqr}. Two methods that do attempt to find the smallest intervals are optimal distributional conformal prediction (DCP) and conformal histogram regression (CHR). 

Optimal DCP uses an estimate of the conditional CDF, which can be estimated using conditional quantile models, $\hat{F}(y, \vect{x}) = \hat{P}(Y \leq y \mid \vect{X} = \vect{x})$. 
The goal of optimal DCP is to find the initial quantiles that lead to the smallest intervals. To do this, define 
\[\hat{Q}(\tau, \vect{x}) = \inf \{y :\hat{F}(y, \vect{x}) \geq \tau \},\] and
\[
\hat{b}(\vect{x}, \alpha) = \argmin_{z \in [0, \alpha]} \hat{Q}(z + 1 - \alpha, \vect{x}) - \hat{Q}(z, \vect{x}),
\] 
which are estimated on the training set.  The non-conformity scores are then,
\[
V_i = |\hat{F}(Y_i, \vect{X}_i) - \hat{b}(\vect{X}_i, \alpha) - \frac{1}{2}(1 - \alpha)|,
\]
leading to a prediction interval of,
\[
\{y: |\hat{F}(y, \vect{X}_{n+1}) - \hat{b}(\vect{X}_{n+1}, \alpha) - \frac{1}{2}(1 - \alpha)| \leq Q_{1 - \alpha}(\vect{V}; \mathcal{Z}_{cal}) \}.
\]
\citep{dcp}. 

Conformal histogram regression is another method that attempts to find the shortest prediction intervals using quantile regression. First, a conditional quantile model is used to build a conditional histogram. See \cite{CHR} Section 2.1 for more on constructing conditional histograms. With the conditional histogram, $T$ nested unconformalized prediction intervals are formed to have coverage $\tau_t = t/T, \> t = 1, \ldots, T.$ The conformity score, $V_i$ is then the smallest $\tau_t$ that contains $Y_i$. The final conformal prediction interval is then the nested interval with coverage $\hat{\tau}$, where $\hat{\tau}$ is the $1-\alpha$ empirical quantile of $\{V\}$ \citep{CHR}.

A method that attempts to find the smallest sets, instead of intervals, using conditional density estimation is HPD-split. The HPD-split method is outlined in Algorithm~\ref{alg:HPD_split}. The non-conformity score used can be seen visually in~\cref{fig:hpd_score}. 

\begin{algorithm}
    \caption{HPD-Split}\label{alg:HPD_split}
    \textbf{Input:} level $\alpha$, data = $\mathcal{Z} = (Y_i, \vect{X}_i)_{i \in \mathcal{I}}$, test point $(\vect{x})$, and conditional density algorithm $\mathcal{B}$ \newline
    \textbf{Procedure:}
    \begin{algorithmic}[1]

    \State Split $\mathcal{Z}$ into a training fold $\mathcal{Z}_{tr} \overset{\Delta}{=} (Y_i, \vect{X}_i)_{i \in \mathcal{I}_{tr}}$ and a calibration fold $\mathcal{Z}_{cal} \overset{\Delta}{=} (Y_i, \vect{X}_i)_{i \in \mathcal{I}_{cal}}$
    \State Fit $\hat{f} = \mathcal{B}(\{ (\vect{X}_i, Y_i): i \in \mathcal{I}_{tr} \})$
    \State Let $\hat{H}$ be an estimate of the cdf of the split residuals, $\hat{f}(Y|\vect{X})$, obtained by numerical integration
    \State Let $U_{\lfloor \alpha \rfloor}$ be the $\lfloor \alpha (n_{cal} + 1) \rfloor$ smallest value of $\{ \hat{H}(\hat{f}(y_i|\vect{x}_i) |\vect{x}_i):i \in \mathcal{I}_{cal} \}$
    \State Build a finite grid over $\mathcal{Y}$ and, by interpolation, \textbf{return}  $\{ y:\hat{H}(\hat{f}(y|\vect{x}) |\vect{x}) \geq U_{\lfloor \alpha \rfloor}\}$

    \end{algorithmic}

\end{algorithm}

\begin{figure}[ht]
    \centering
\includegraphics[scale = 0.7]{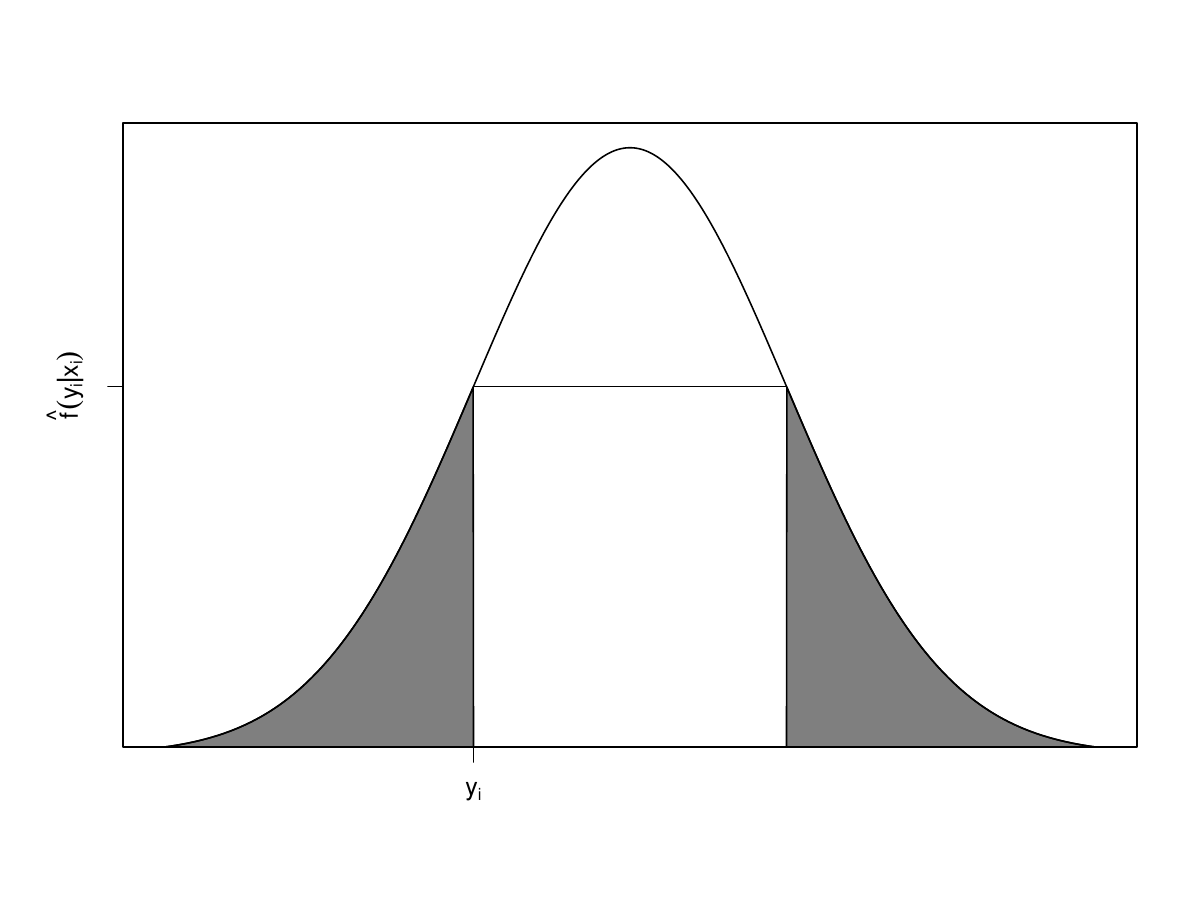}
\caption{The HPD-split score for a sample $(y_i, \vect{x}_i)$ is the shaded region of the plot.}\label{fig:hpd_score}
\end{figure}

A visual of how the CHCDS score works can be seen in~\cref{fig:our_score}. 

\begin{figure}[ht]
    \centering
\includegraphics[scale = 0.7]{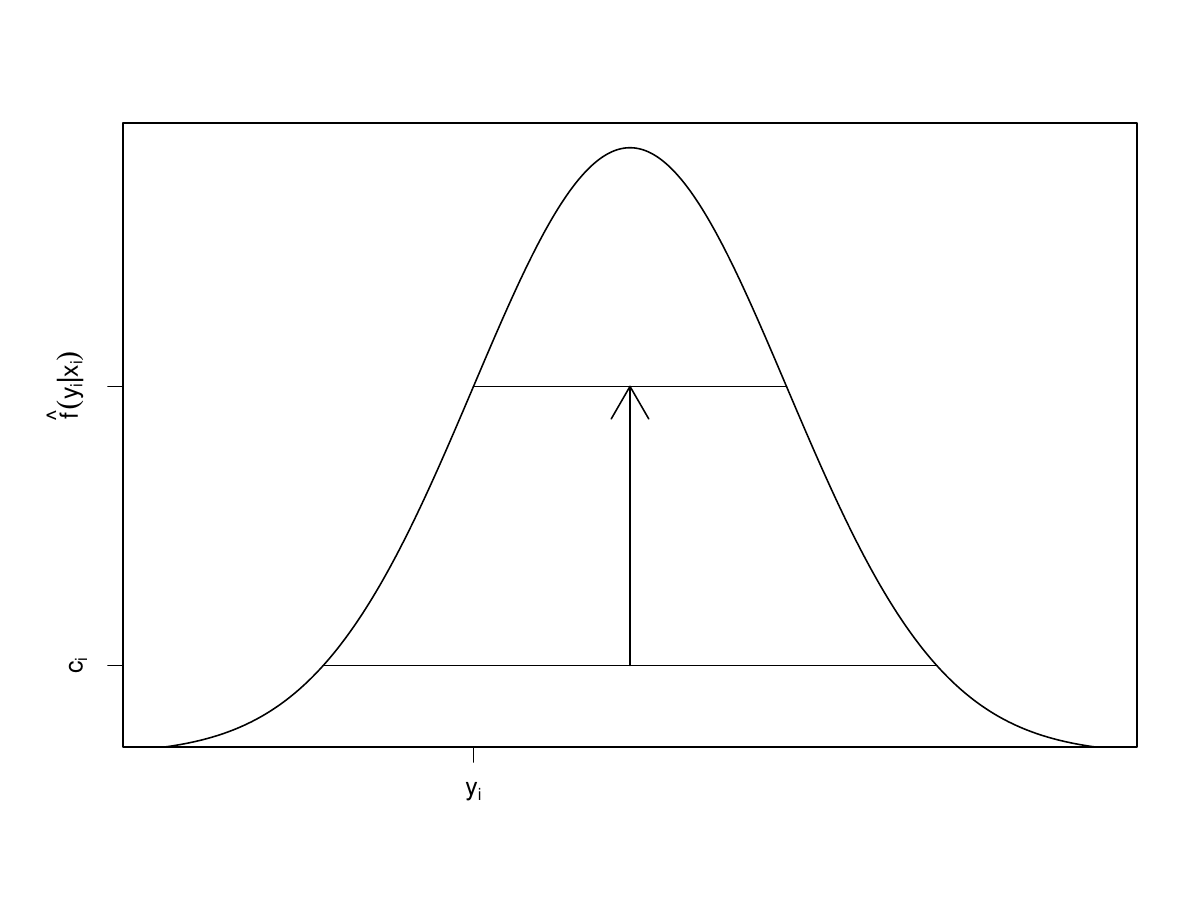}
\caption{The CHCDS score for a sample $(y_i, \vect{x}_i)$ is shown by the arrow.}\label{fig:our_score}
\end{figure}

\section{CHCDS Comparison} 
\subsection{Conditional Density Score vs CHCDS}
\label{sec:density_vs_chcds}
A small simulation is given below to demonstrate the advantage of CHCDS vs the negative density score. The data were generated as follows: \\$X \sim \text{Unif}(-1.5, 1.5)$, \> $Y \mid X \sim \mathcal{N}(5 + 2X, |X| + 0.05)$ \\
In the following simulation, the true conditional density was used, 500 data points were used to calibrate each score, the simulation size was 10,000, and the coverage rate was set to $1 - \alpha = 0.90$. The difference in conditional coverage can be seen in~\cref{fig:CHCDS_neg_dens_conditional}. Clearly, although the coverage and set sizes are similar (see \cref{tab:CHCDS_neg_dens}), CHCDS has much better conditional coverage. 

\begin{table}[ht]
\begin{center}
\caption{CHCDS vs Negative Density Comparison}
    \begin{tabular}{lcc}
         Approach &  Coverage & Set Size   \\
        CHCDS & 897 (3) & 2597 (1) \\
        Negative Density & 899 (3) & 2456 (1) \\
    \end{tabular}
    \label{tab:CHCDS_neg_dens}
\end{center}
    \caption{
	Monte Carlo error given in parentheses only if it is greater than 0.001. All values have been
multiplied by $10^3$.
}\end{table}

\begin{figure}[ht]
    \centering
\includegraphics[scale = 0.3]{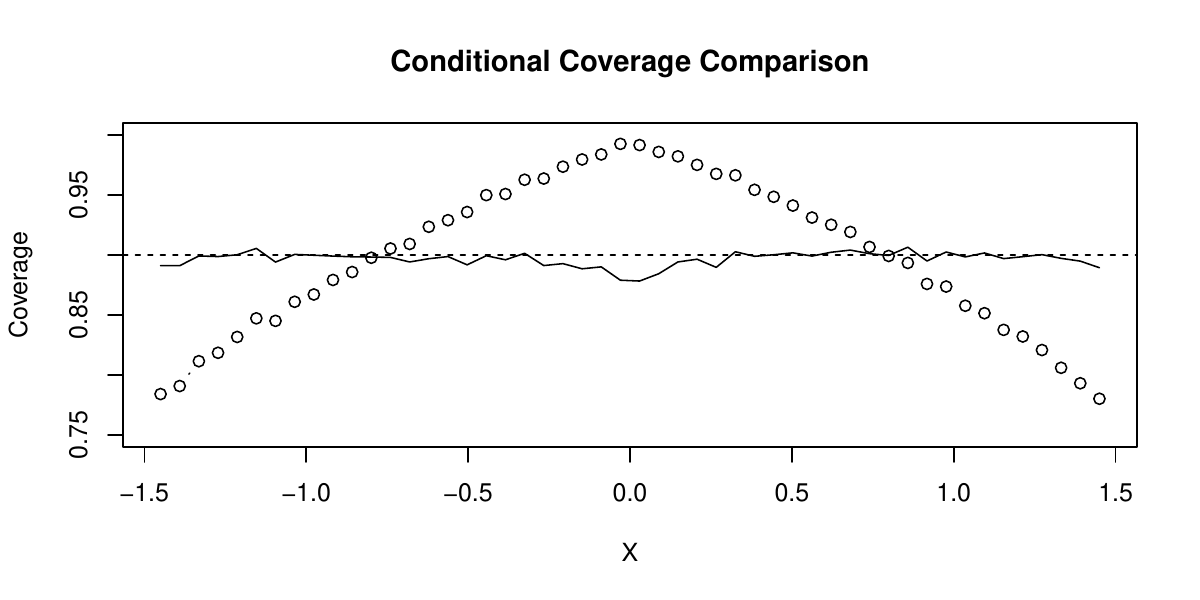}
\caption{A plot showing the conditional coverage of CHCDS (solid line) and the negative density score (open circles) at a given value of X. The dashed line represents the desired 90\% coverage.}\label{fig:CHCDS_neg_dens_conditional}
\end{figure}

\subsection{Multiplicative CHCDS}
\label{sec:division_subtraction_comparison}
A problem can occur when the prediction level is high and the conditional density estimator overestimates the estimated density value cutoff for the unadjusted $1 - \alpha$ prediction set (that is, the estimated density under covers before the conformal adjustment). When this occurs, the final density cutoff for the new prediction point is a negative value, or an infinite prediction set. To solve this issue, we propose a new method below that has the same properties and benefits as CHCDS without this potential drawback.

Let the new conformity scores be
\[
V_{i} = \hat{f}(Y_i\mid \vect{X}_i) / \hat{c}(\vect{X}_i) \text{, } \forall i \in \mathcal{I}_{cal}.
\]

Then, the final conformal prediction set will be
\[
\vect{C}(\vect{X}_{n+1}) = \{y: \hat{f}(y\mid \vect{X}_{n+1}) > \hat{c}(\vect{X}_{n+1}) \times \hat{q} \},
\]
where $\hat{q} =\lfloor \alpha (n_{cal} + 1) \rfloor$th smallest value of $\{ \vect{V} \}$. When the conditional density has a high variability, so the density cutoff values are small, we can add a small constant, $\gamma$ for numerical stability. 

\[
V_{i} = \hat{f}(Y_i\mid \vect{X}_i) / (\hat{c}(\vect{X}_i) + \gamma) \text{, } \forall i \in \mathcal{I}_{cal}.
\]

Then, the final conformal prediction set will be
\[
\hat{\vect{C}}(\vect{X}_{n+1}) = \{y: \hat{f}(y\mid \vect{X}_{n+1}) > (\hat{c}(\vect{X}_{n+1}) + \gamma) \times \hat{q} \},
\]
where $\hat{q} =\lfloor \alpha (n_{cal} + 1) \rfloor$th smallest value of $\{ \vect{V} \}$. 

We include a comparison of CHCDS-subtraction with this method for some of the simulation scenarios. The results in all cases are nearly identical. We also include a toy simulation example to demonstrate how CHCDS can cause infinite prediction sets. 

Below are two plots comparing additive conformal adjustment versus multiplicative conformal adjustment, referred to as 
CHCDS-subtraction and CHCDS-division,  
with different density estimators in the mixture scenario with the same simulation size and sample sizes as in~\cref{sec:chcds_simulation}. We set the coverage rate to be $1 - \alpha = 99\%$ to see how well both methods performed in an edge case. In~\cref{fig:div_sub_comp_kernel} and~\cref{fig:div_sub_comp_knn} the blue line is the conditional coverage for CHCDS-subtraction and the black line is the conditional coverage for CHCDS-division. In this scenario, we can see they are nearly identical. In both cases, the subtraction method did not output an infinite prediction set.

\begin{figure}[ht]
    \centering
\includegraphics[scale = 0.5]{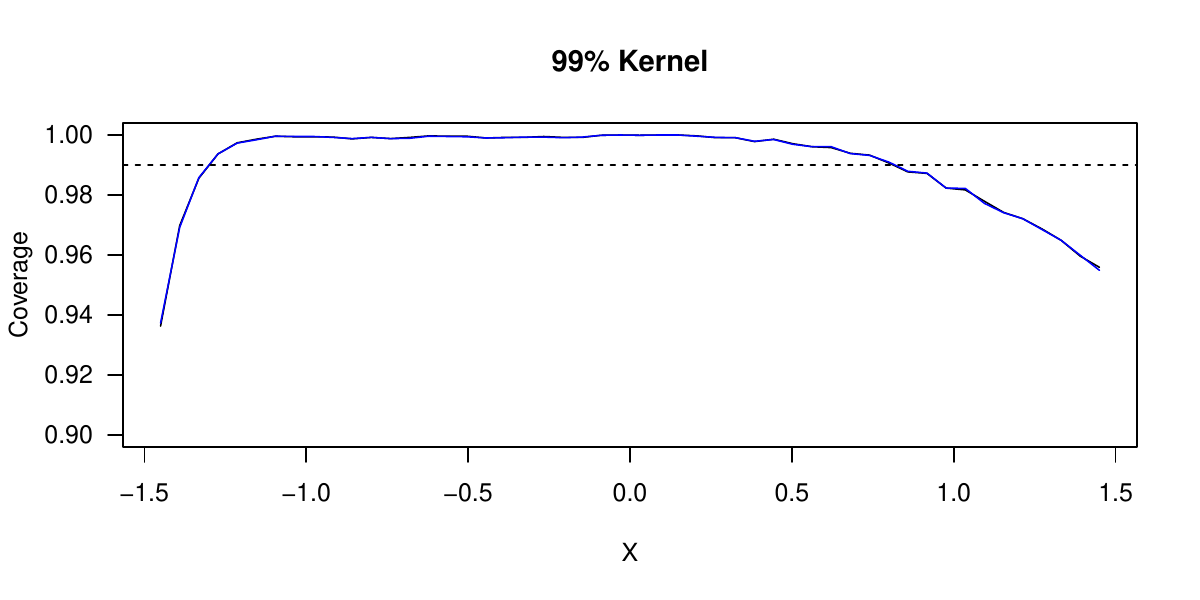}
\caption{Comparison of conditional coverage between CHCDS-subtraction Kernel (blue) and CHCDS-division Kernel (black) in the mixture scenario. The dashed line shows the desired 99\% coverage.}
\label{fig:div_sub_comp_kernel}
\end{figure}

\begin{figure}[ht]
    \centering
\includegraphics[scale = 0.5]{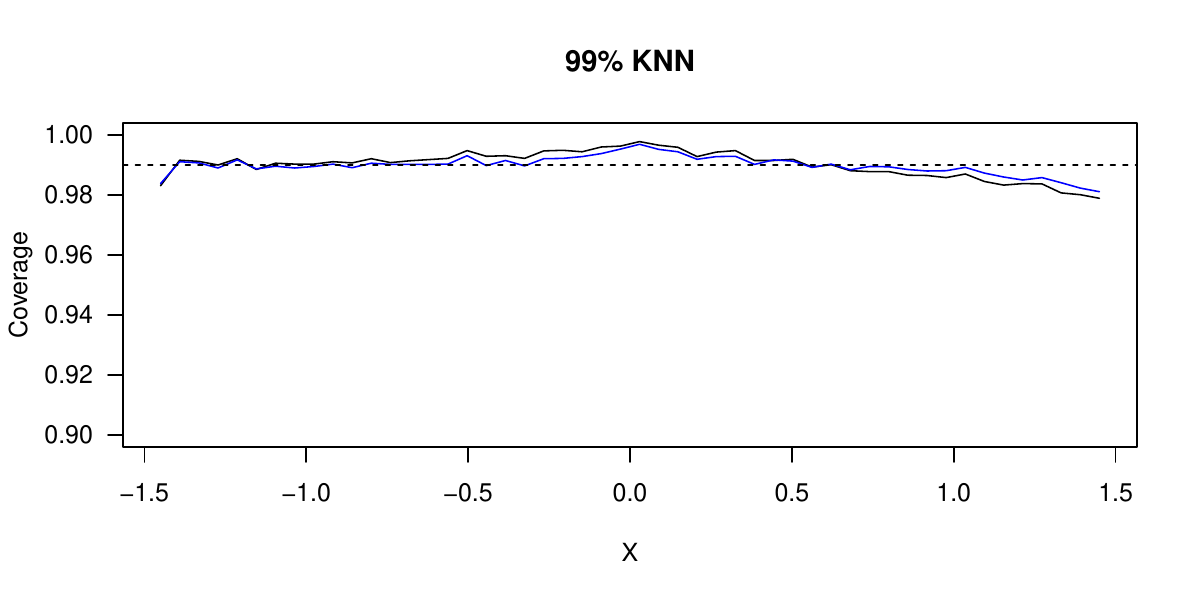}
\caption{Comparison of conditional coverage between CHCDS-subtraction (KNN) (blue) and CHCDS-division (KNN) (black) in the mixture scenario. The dashed line shows the desired 99\% coverage.}
\label{fig:div_sub_comp_knn}
\end{figure}

To show the potential failure of CHCDS-subtraction, we looked at a heteroskedastic example with data generated in the following way. $X \sim \text{Unif}(-5, 5)$, $Y\mid X \sim \mathcal{N}(0, | X|  + 0.01$. A scatterplot of the relationship between the response and covariate can be found in~\cref{fig:div_sub_comp_normal_scatterplot}. Data were split, 1000 for training and 500 for calibration. The simulation size was 10,000. We only used a conditional kernel density estimator. The target coverage rate was $1 - \alpha = 99\%$, but we found the $98.5\%$ unconformalized estimated highest density cut-off. This, along with the heteroskedastic nature of the data, was to ensure we would need to lower the cut-off with our conformal adjustment, while having high variability conditional densities (and, thus, density cut-off values near 0). The comparison of the method's conditional coverages can be found in~\cref{fig:div_sub_comp_normal_kernel}. CHCDS-subtraction gave infinite prediction intervals in 1.13\% of the out of sample test cases. As expected, these nearly all ocurred in the regions with large variability, which correspond to large values of $|X| $. CHCDS-division never gave an infinite prediction interval.  

\begin{figure}[ht]
    \centering
\includegraphics[scale = 0.5]{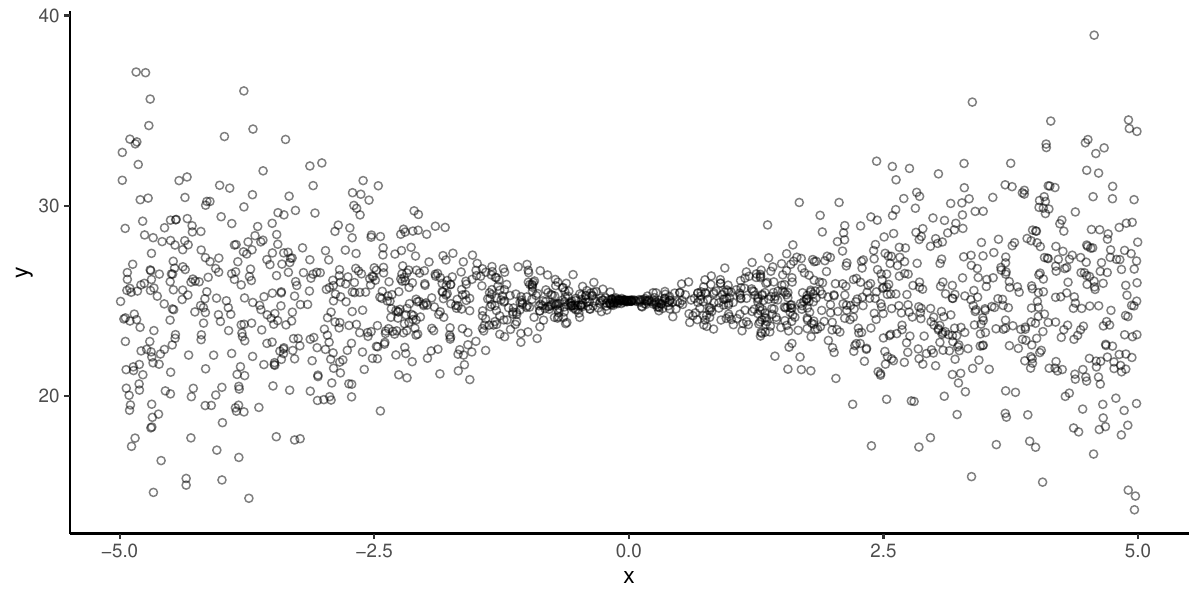}
\caption{A scatterplot of the heteroskedastic Normal distribution response covariate relationship. }
\label{fig:div_sub_comp_normal_scatterplot}
\end{figure}

\begin{figure}[ht]
    \centering
\includegraphics[scale = 0.5]{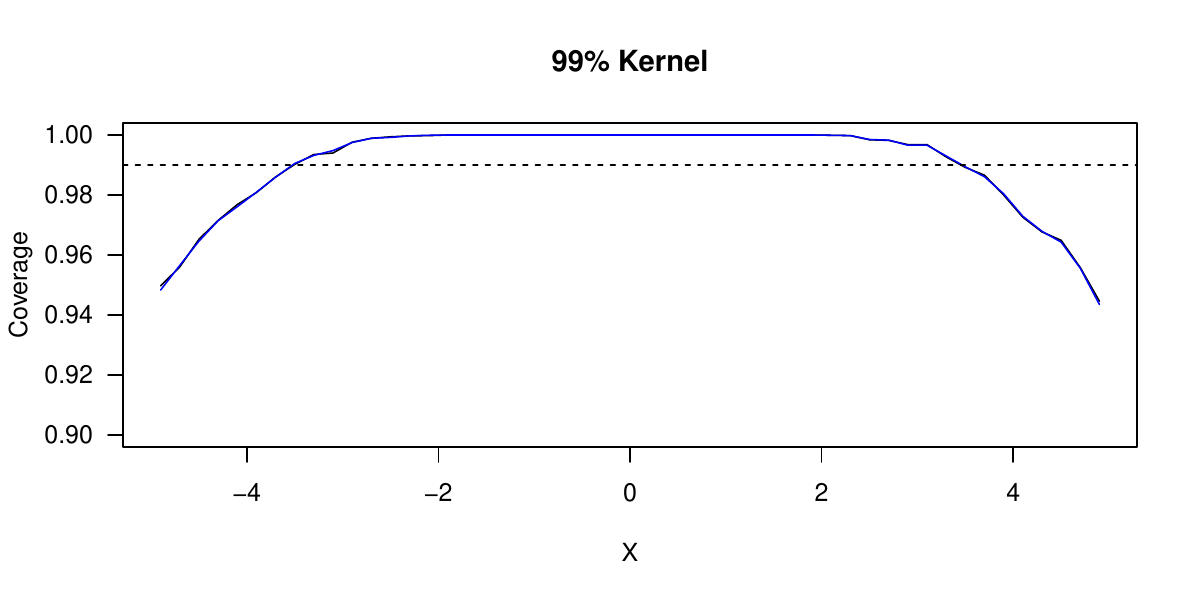}
\caption{Comparison of conditional coverage between CHCDS-subtraction Kernel (blue) and CHCDS-division Kernel (black) in the heteroskedastic Normal scenario. The dashed line shows the desired 99\% coverage.}
\label{fig:div_sub_comp_normal_kernel}
\end{figure}

\section{Conditional Density Estimation}~\label{sec: conditional density estimation}

In this section, we describe some existing methods for conditional density estimation. 
Denote the conditional density estimator as $\hat{f}(Y\mid \vect{X})$, where $\vect{X}$ are the $d$ conditioning variables. 
The simplest method for conditional density estimation comes from \cite{HiranoandImbens2004}, where they used it to estimate generalized propensity scores. They assumed that 
\[
Y\mid \vect{X} \sim \mathcal{N}\{\mu(\vect{X}), \sigma^2\}.
\]
They then estimated the function $\mu$ and $\sigma^2$ with $\hat{\mu}$ and $\hat{\sigma}^2$ using lest squares regression. Using these estimations, the conditional densities were computed as

\[
\hat{f}(y\mid \vect{X}) = (2\pi\hat{\sigma}^2)^{-1/2}
\text{exp}[-\{y - \hat{\mu}(\vect{X})\}^2/ (2\hat{\sigma}^2) ]. 
\]

A non-parametric approach to density estimation is kernel density estimation. It was originally a form of unconditional density estimation \citep{rosenblatt_density_1956, parzen_density_1962}. Let $K(y)$ be the kernel function chosen. Assume that $K(y)$ is symmetric, $\int_{-\infty}^{\infty} y^2K(y) dy < \infty$, and $\int_{-\infty}^{\infty} K(y)dy = 1$. Denote $h$ as the bandwidth, or smoothing parameter, of the chosen kernel density estimator such that $h \to 0$ and $nh \to \infty$ as $n \to \infty$. The bandwidth $h$ is commonly chosen to be a function of $n$. Then, the kernel density estimator is 
\[
\hat{f}(y) = (nh)^{-1} \sum\limits_{i = 1}^nK\{(Y_i - y)/h\}.
\]

\cite{hyndman_conditional_density_1996} extended this idea to kernel conditional density estimation. Let $f(y, \vect{x})$ denote the joint density of $(Y,\vect{X})$ and $f(\vect{x})$ denote the marginal density of $\vect{X}$. Then, we can write the conditional density of $Y\mid (\vect{X} = \vect{x})$ as $f(y\mid \vect{x}) = f(y,\vect{x}) / f(\vect{x})$. The kernel conditional density estimator is then
\begin{equation}
    \hat{f}(y\mid \vect{x}) = \hat{f}(y, \vect{x})/\hat{f}(\vect{x}), \label{eq: KDE}
\end{equation}
where
\[
\hat{f}(y, \vect{x}) = (n\times b\times \prod_{i=1}^n a_i)^{-1}\sum\limits_{j = 1}^n  K(b^{-1}\| y - Y_j\|_y)\times \prod_{i=1}^d K(a_i^{-1}\| {x}_i - {X}_{ij}\|_{\vect{x}}) 
\]
and
\[
\hat{f}(\vect{x}) = (n\times  \prod_{i=1}^n a_i)^{-1}\sum\limits_{j = 1}^n  \prod_{i=1}^d K(a_i^{-1}\| {x}_i - {X}_{ij}\|_{\vect{x}}) ,
\]
where $\vect{X}_i = (X_{i1}, \ldots, X_{id})^\intercal$. Here, $\|\cdot \| _{\vect{x}}$ and $\|  \cdot \|_y$ are distance metrics, for example the Euclidean distance. The kernel used satisfies the same properties as those used in unconditional kernel density estimation, and $a_i$ and $b$ are the bandwidth parameters \citep{rosenblatt_multivariate_kernel_1971, cacoullos_multivariate_kernel_1966}, which are often taken to be identically equal to $h$ if the the variables are of similar scale.

Kernel density estimation gives more influence to points closer to the point of interest, but that influence can be increased by computing the density estimates on the nearest $k$ neighbors, instead of on all $n$ data points. 
This can work well if we have a large number of data points and the true conditional density varies heavily depending on the covariates, for example a mixture distribution \citep{izbicki_2018_NNCDE}. Instead of a vector for the covariate's bandwidth, a bandwidth matrix along with a multivariate kernel can be used \citep{mack_rosenblatt_1979_knn_kernel}. 

Another combination of k-nearest neighbors and kernel density estimation replaces the bandwidths, $h$, $a$, and $b$, with the distance from the $k$th nearest neighbor to $\vect{x}$ or $y$ \citep{mack_rosenblatt_1979_knn_kernel, moore_yackel_consistency_knn_kernel_1977, loftsgaarden_quesenberry_knn_kernel_intro_1965}. This allows for unequal weighting of the observations. There are many other ways of choosing kernel density bandwidths and automatically selecting bandwidths that can be found in \cite{chiu_review_bandwidth_1996, heidenreich_review_bandwidth_2013}. 

FlexCode is another non-parametric density estimator. First, the response is scaled to be in the interval $[0, 1]$. Then, FlexCode uses a basis expansion to model the density. One option for the basis is the Fourier basis:
\begin{align*}
\phi_1(z) = 1; \qquad \phi_{2i + 1}(z) = 2^{1/2}\sin(2\pi i z) \qquad \phi_{2i}(z) = 2^{1/2}\cos(2\pi i z), \> i \in \mathbb{N}.
\end{align*}
For a fixed $\vect{x}$, as long as $\int_{-\infty}^{\infty }|f(y\mid \vect{x})|^2dy < \infty$, the conditional density can be written as 
\[
f(y\mid \vect{x}) = \sum\limits_{i \in \mathbb{N}} \beta_i(\vect{x}) \phi_i(y),
\]
where $\beta_i(\vect{x}) = \mathbb{E}(\phi_i(Y)\mid \vect{x})$ \citep{izbicki_lee_nonpar_cde_2016}. Now, define the FlexCode conditional density estimator as
\[
\hat{f}(y\mid \vect{x}) = \sum\limits_{i = 1}^I\hat{\beta}_i(\vect{x}) \phi_i(y),
\]
where $\hat{\beta}_i(\vect{x})$ is estimated using a conditional mean algorithm (for example, a random forest estimator) and $I$ is a tuning parameter that controls the bias-variance tradeoff \citep{flexcode_izbicki_2017}. 

Mixture models are a parametric approach to finding an estimate of the conditional density \citep{bishop_1994_mixture_density}. Let the joint density be represented as 
\[
f_{\vect{\theta}}(y, \vect{x}) = \sum_{i = 1}^K \pi_i f_{\vect{\psi}_i}(y, \vect{x}),
\]
where $K$ represents the number of mixture components, $\pi_j$ represents the weight of the $j$th mixture component, and $\vect{\psi}_j$ represents the parameters of the $j$th mixture component. For example, the mean vector and covariance matrix for a Normal distribution \citep{gaussian_mix_R_2023}. These parameters can be estimated using an iterative process, for example an EM algorithm or a deep neural network \citep{Viroli2019_gaussian_mixture}. 

We can use the same idea to estimate the marginal density of $\vect{X}$, 
\[
{g}_{{\vect{\zeta}}}(\vect{x}) = \sum_{i = 1}^J \pi_i g_{\vect{\eta_i}},
\] 
where $J$ represents the number of mixture components, $\pi_j$ represents the weight of the $j$th mixture component, and $\vect{\eta}_j$ represents the parameters of the $j$th mixture component. Combining the joint and marginal estimates, the conditional mixture density estimator is 
\[
\hat{f}(y\mid \vect{x}) = {f}_{\hat{\vect{\theta}}}(y, \vect{x})/{g}_{\hat{\vect{\zeta}}}(\vect{x}).
\]

A second approach to conditional density estimation with mixture distributions is to assume that the conditional distribution of $Y\mid \vect{X}$ is the convolution of $K$ distributions. 
\[
f_{\vect{\theta}}(y\mid \vect{x}) = \sum_{i = 1}^K \pi_i f_{\vect{\psi}_i}(y\mid \vect{x}),
\]
where $K$ represents the number of mixture components, $\pi_j$ represents the weight of the $j$th mixture component, and $\vect{\psi}_j$ represents the parameters of the $j$th mixture component. 
Though the true density may not be a mixture of, for example, Normal distributions, a Gaussian mixture distribution can provide a reasonable approximation.

Normalizing flows are another method that can be used for conditional density estimation. Start with the idea of a change of variable transformation. Let $y = g(u)$, where $g$ is a bijective and differentiable function and $\pi(\cdot)$ is the density of $U$. The density of $Y$ can then be represented as 
\[
f_Y(y) = \pi\{g^{-1}(y)\} \left|  \text{det} \frac{\partial g}{\partial y\hfill}^{-1}(y)\right| .
\]
The base density, which is typically chosen to be easy to evaluate for any input (a common choice is a standard Gaussian density) \citep{Papamakarios_normalizingflows_2017}. 

Extending this idea to a conditional distribution is fairly straightforward. Let $y = g(u, \vect{x})$, where $g$ is a bijective and differentiable function and $\pi(\cdot\mid \vect{x})$ is the conditional density of $U\mid \vect{x}$. The conditional density of $Y\mid \vect{x}$ can then be represented as
\[
f_{Y\mid \vect{X}}(y\mid \vect{x}) = \pi\{g^{-1}(y, \vect{x})\mid \vect{x}\} \left|  \text{det} \frac{\partial g}{\partial y\hfill}^{-1}(y, \vect{x})\right| .
\]
As in the unconditional case, the base density is typically chosen to be easy to evaluate \citep{winkler_conditionalflows_2023}. For example, we might have $U\mid X \sim \mathcal{N}(0, |X|^2)$ and $Y\mid X \sim \mathcal{N}(X, |X|^2)$. Then our function, $g(u, x) = u + x = y$, $g^{-1}(y, x) = y - x$, and $|\partial g^{-1}(y, x)/\partial y| = 1$ \citep{kingma_ARflow_2016}. 

When doing conditional density estimation, we typically do not know what the function $g$ is, or even if we can find a suitable transformation using only one function. Instead, let $y = g_{K}(g_{K-1} (\ldots g_1(u_0, x)))$. The log-likelihood for $f_{Y\mid \vect{x}}(y, \vect{x})$ can then be written as 
\begin{equation}\label{eq:NF_loss_function}
f_{Y\mid \vect{x}}(y, \vect{x}) = \log \{\pi(u_0\mid \vect{x})\} - \sum_{i = 1}^K \log\Big\{\left | \det \frac{\partial g_i}{\partial u_{i - 1}\hfill}(u_{i-1}, \vect{x}) \right|  \Big\},
\end{equation}
where $u_i = g_i(u_{i-1}, \vect{x})$ \citep{trippe2018normalizingflow, rezende_normalizingflow_2015}. Knowing this gives us a function to maximize, so the negative of \eqref{eq:NF_loss_function} becomes our objective loss function. The functions, $g$, typically have parameters that need to be estimated. In the earlier example of $g(u, x) = u + x$, the function might actually be $g(u, x) = u + \mu(x)$, where we do not know $\mu(\cdot)$. \citep{Papamakarios_normalizingflows_2017, winkler_conditionalflows_2023}. We can use neural networks (or other tools) along with \eqref{eq:NF_loss_function} to approximate the $g$ functions. Combining these estimates gives us an estimate for $\hat{f}_{Y\mid \vect{X}}(y\mid \vect{x})$. There are many classes of $g$'s that can be used in normalizing flows, see \cite{winkler_conditionalflows_2023, trippe2018normalizingflow, kobyzev_2021_nf_functions, Papamakarios_2021_nf_functions} for examples.  

\subsection{Unimodal Density Estimation \& Comments on Multi-modal Densities}\label{sec:unimodal density estimation}

In this subsection, we introduce a few methods of density estimation that ensure the estimated density is unimodal. The first is to assume that the conditional distribution of $Y \mid \vect{X}$ can be modeled by a unimodal parametric density, for example a Normal distribution. The parameters can be estimated in many ways, including a deep neural network or an iterative process, just as with mixture densities. 

For a non-parametric unimodal density, one approach is to change the weights used in kernel density estimation for each data point from $1/n$ to $p_i$, where $p_i$ is constrained to guarantee a single mode \citep{Hall_Heckman_2002_weighted_kernel_unimodal, hall2002unimodal}. A more common approach that can be used with most non-parametric density estimators to ensure they are unimodal is data sharpening \citep{wolters_unimodal_density_2012,unimodal_kernel_datasharp_Hall_Kang_2005, hall_choi_sharpening}. Data sharpening slightly alters the observed data to ensure a unimodal density. Under mild conditions, when the true density is unimodal, existing data sharpening methods only move points that form spurious modes in the tails closer towards the true mode \citep{unimodal_kernel_datasharp_Hall_Kang_2005}. Let $\vect{x}$ be the observed data and $\vect{y}$ be a possible sharpened data vector. The goal is the minimize $\delta(\vect{x}, \vect{y})$, where $\delta$ is an objective function that measures how different the observed and sharpened data are. For example if the data are univariate, \[
\delta(\vect{x}, \vect{y}) = \sum\limits_{i = 1}^n |x_i - y_i|.
\]
A more thorough review of sharpening the data and different objective functions can be found in \cite{wolters_unimodal_density_2012}. A discussion of data sharpening to ensure a unimodal conditional density can be found in section 4.3 of \cite{wolters_unimodal_density_2012}. 

Another approach to unimodal density estimation is Bernstein polynomials presented as a mixture of Beta kernels. Consider the following class of density estimators,
\[
\hat{f}(x; m, \vect{\omega}) = \sum\limits_{k = 1}^m \omega_k f_b(x; k, m - k + 1), \> k = 1, \ldots, m,
\]
where $\vect{\omega}$ is a vector of weights and $f_b(\cdot)$ is the Beta density function with shape parameters $k$ and $m - k + 1$. As long as the weights are non-negative, sum to 1, and 
$\omega_1 \leq \omega_2 \leq \cdots, \leq \omega_{k^*} \geq \omega_{k^* + 1} \geq \cdots \geq \omega_m$ the density estimator will be unimodal. This approach to unimodal density estimation can be extended to more general supports of $(a, b)$ by using the linear transformation, $u = (x - a) / (b - a)$ \citep{turnbull_ghosh_bernstein_unimodal_2014}.

When there is a lack of substantive knowledge that the true predictive distribution is unimodal, these methods should not be used. When a density estimate is multi-modal, it can be indicative that the true random error has a multi-modal distribution. This is often caused by a missing covariate or a multifunctional covariate-response relationship. For example, the heights of individuals based on age with gender as a missing covariate or traffice speed-flow data. When this is the case, descriptive models and summaries should be carefully chosen to ensure that the data and covariate-response relationship are sufficiently described \citep{hyndman_conditional_density_1996,einbeck2006moderegression, chen_etal_modal_2016, chen_mode_regression2018}.

\section{Sufficient conditions for assumption~\ref{assumption: C3}}~\label{sec:validity of C3}
The validity of assumption~\ref{assumption: C3} requires technical conditions on both the true density function and the kernel function.
Given assumptions \ref{assumption: C1} and \ref{assumption: C2}, and that the joint pdf $f(y, \vect{x})$ is $\zeta$-H\"older continuous, where $0<\zeta\le 1$, (i.e., there exists a constant $K'$ such that $|f(y, \vect{x}) - f(y', \vect{x}')|\le K' (|y-y'|+|\vect{x}-\vect{x}'|)^\zeta$ for all $(y, \vect{x})$ and $(y', \vect{x}')$). Consider  kernel density estimation defined by \eqref{eq: KDE} with a spherically symmetric kernel of compact support, using Euclidean distance and with identical bandwidth  equal to $h$. Let 
$b_{n}= h^\zeta+\{\log n/(n h^{d+1})\}^{1/2}$. Suppose $h > \left(\log n/n\right)^{1/(d+1)} $ and $h\to 0$ as $n\to \infty$. 
Under the preceding conditions, \citep[Theorem 2]{jiang2017_uniformkde_convergence} implies  the validity of assumption~\ref{assumption: C3}. This is because it follows from \citep[Theorem 2]{jiang2017_uniformkde_convergence} that $| \hat{f}(y,\vect{x})- {f}(y,\vect{x})| $ is  less than a fixed multiple of $b_n$ for all $\vect{x}$, except for an event of probability less than $1/n$, and the same holds for $| f(\vect{x})-\hat{f}(\vect{x})| $.   The convergence regarding $\hat{f}(\cdot\mid \vect{x})$, as stated in assumption~\ref{assumption: C3}, then follows from  the inequality
\begin{eqnarray*}
    &&| \hat{f}(y\mid \vect{x})-f(y\mid \vect{x})|  \\
    &\le & \{| \hat{f}(y,\vect{x})- {f}(y,\vect{x})| +f(y\mid \vect{x})| f(\vect{x})-\hat{f}(\vect{x})| \}/\hat{f}(\vect{x}).
\end{eqnarray*}
Note that for all sufficiently small $h$, the support of the kernel density estimator $\hat{f}(y, \vect{x})$ lies within some fixed compact set, hence the claim concerning the uniform convergence of $\hat{F}(\cdot\mid x)$ can be readily verified.

\section{Proofs}~\label{sec: proofs}

\noindent The next two lemmas come from \cite{romanocqr, conformal_shift, split_conformal_lei_2016, conformal_book}. 

\begin{lemma} \label{lemma:quantiles_exchangeability_known} (Quantiles and exchangeability). Suppose $Z_1, \ldots Z_n$ are exchangeable random variables.

\noindent For any $\alpha \in (0, 1)$,
\[
pr\{(Z_n \leq \hat{R}_n(\alpha)\} \geq \alpha,
\]

\noindent where $\hat{R}_n$ is the empirical quantile function, $\hat{R}_n(\alpha) = Z_{(\lceil\alpha n\rceil)}$.

\noindent Moreover, if the random variables $Z_1, \ldots, Z_n$ are almost surely distinct, then
\[
pr\{Z_n \leq \hat{R}_n(\alpha)\} \leq \alpha + \frac{1}{n}.
\]

\end{lemma}

\begin{lemma} \label{lemma:inflation_quantiles_known} (Inflation of Quantiles). Suppose $Z_1, \ldots Z_n$ are exchangeable random variables.

\noindent For any $\alpha \in (0, 1)$,
\[
pr[Z_{n+1} \leq \hat{R}_n\{(1 + n^{-1})\alpha\}]\geq \alpha.
\]
\noindent Moreover, if the random variables $Z_1, \ldots, Z_n$ are almost surely distinct, then
\[
pr\{Z_{n+1} \leq\hat{R}_n\{(1 + n^{-1})\alpha)\} \leq \alpha + (n + 1)^{-1}.
\]

\end{lemma}

\noindent The following two lemmas are similar to those from \cite{romanocqr, conformal_shift, split_conformal_lei_2016, conformal_book}. The following two lemmas are standard in showing the validity of the unconditional coverage guaranteed when using conformal prediction.

\begin{lemma} \label{lemma:quantiles_exchangeability_ours} Suppose $Z_1, \ldots Z_n$ are exchangeable random variables.

\noindent For any $\alpha \in (0, 1)$,
\[
pr\{Z_n > \hat{Q}_n(\alpha)\} \geq 1 - \alpha,
\]

\noindent where $\hat{Q}_n$ is the empirical quantile function, $\hat{Q}_n(\alpha) = Z_{(\lfloor\alpha n\rfloor)}$.
\noindent Moreover, if the random variables $Z_1, \ldots, Z_n$ are almost surely distinct, then
\[
pr(Z_n > \hat{Q}_n(\alpha)\} \leq 1 - \alpha + n^{-1}.
\]

\end{lemma}

\begin{proof}
For the lower-bound, by exchangeability of $Z_1, \ldots, Z_n$, $pr\{Z_i > \hat{Q}_n(\alpha)\}= pr\{Z_n > \hat{Q}_n(\alpha)\}$, for all $ i$. Therefore,
\[
E[1 - \hat{F}_n\{\hat{Q}_n(\alpha)\}] = \sum\limits_{i = 1}^n n^{-1}pr\{Z_i > \hat{Q}_n(\alpha)\} = pr\{Z_n > \hat{Q}_n(\alpha)\},
\]
\noindent where $\hat{F}_n(z) := n^{-1}\sum_{i=1}^n \one(Z_i \leq z)$. 

\noindent By the definition of the empirical cdf, $1 - \hat{F}_n\{\hat{Q}_n(\alpha)\} \geq 1 - \alpha$.

\noindent Taking the expectation, we get, 
\[
pr\{Z_n > \hat{Q}_n(\alpha)\} \geq 1 - \alpha,
\]
as desired.

\noindent For the upper-bound, apply~\cref{lemma:quantiles_exchangeability_known} with $\hat{R}_n(\alpha) = Z_{(\lceil\alpha n - 1 \rceil)}$ instead of $Z_{(\lceil\alpha n \rceil)}$. 
Then we have that 
\[
\alpha - n^{-1} \leq pr\{Z_n \leq Z_{(\lceil\alpha n - 1 \rceil)}\}.
\]
Notice that $Z_{(\lceil\alpha n - 1 \rceil)} \leq Z_{(\lfloor\alpha n \rfloor)}$, so
\[
1 - (\alpha - n^{-1}) \geq pr\{Z_n > Z_{(\lceil\alpha n \rceil - 1)}\} \geq pr\{Z_n > \hat{Q}_n(\alpha)\}.
\] 
So,
\[
1 - \alpha + n^{-1} \geq pr\{Z_n > \hat{Q}_n(\alpha)\}.
\]
\end{proof}

\begin{lemma} \label{lemma:inflation_quantiles_ours} Suppose $Z_1, \ldots Z_n$ are exchangeable random variables.

\noindent For any $\alpha \in (0, 1)$,
\[
pr[Z_{n+1} > \hat{Q}_n\{(1 + n^{-1} )\alpha\}] \geq 1 - \alpha.
\]
    
\end{lemma}

\begin{proof}
Let $Z_{(k, m)}$ denote the $k$th smallest value in $Z_1, \ldots, Z_m$. 

\noindent For any $0 \leq k \leq n$, we have
\[
Z_{n+1} > Z_{(k, n)} \text{ if and only if } Z_{n+1} > Z_{(k, n + 1)}
\]

\noindent First, assume $Z_{n+1} > Z_{(k, n)}$, then $Z_{(k, n+1)} = Z_{(k, n)}$, so $Z_{n+1} > Z_{(k, n + 1)}$.

\noindent Second, assume $Z_{n+1} > Z_{(k, n + 1)}$, then $Z_{(k, n)} = Z_{(k, n+1)}$, so $Z_{n+1} > Z_{(k, n)}$.

\noindent Now, because $\hat{Q}_n\{(1 + n^{-1})\alpha\} = Z_{(\lfloor\alpha(n+1)\rfloor, n)}$ and $\hat{Q}_{n+1}(\alpha) = Z_{(\lfloor \alpha(n+1)\rfloor, n + 1)}$, we have

\[
Z_{n+1} >\hat{Q}_n\{(1 + n^{-1})\alpha\} \text{ if and only if } Z_{n+1} > \hat{Q}_{n+1}(\alpha).
\]
So, 
\[
pr[Z_{n+1} >\hat{Q}_n\{(1 + n^{-1})\alpha\}]= pr\{ Z_{n+1} > \hat{Q}_{n+1}(\alpha)\} \geq 1 - \alpha,
\]
and
\[
pr[Z_{n+1} >\hat{Q}_n\{(1 + n^{-1})\alpha\}]= pr\{ Z_{n+1} > \hat{Q}_{n+1}(\alpha)\} \leq 1 - \alpha + (n + 1)^{-1},
\]
where the last inequality follows by applying~\cref{lemma:quantiles_exchangeability_ours} with $n = n+1$.
\end{proof}

\subsection{Proof of Theorem~1}

\begin{proof}
Let $V_{n+1}$ be the score at the test point $\vect{X}_{n+1}$ and  $k$ denote the smallest $\lfloor(n_{cal} + 1)(\alpha)\rfloor$ value in $\{V_i; i \in \mathcal{I}_{cal}\}$.

\noindent By construction of the prediction interval, we have that
\noindent $V_{n+1} > k$ is equivalent to
\begin{eqnarray*}
 &&\hat{f}(Y_{n+1}\mid \vect{X}_{n+1}) - c(\vect{X}_{n+1}) > k   \\
 &\iff& \hat{f}(Y_{n+1}\mid \vect{X}_{n+1}) > k + c(\vect{X}_{n+1})\\
 &\iff& Y_{n+1} \in \{y: \hat{f}(y\mid \vect{X}_{n+1}) > k + c(\vect{X}_{n+1})\}.
\end{eqnarray*}

\noindent So,
\[
Y_{n+1} \in \hat{C}(\vect{X}_{n+1}) \iff V_{n+1} > k. 
\]

\noindent We then have that
\[
pr\{Y_{n+1} \in \hat{C}(\vect{X}_{n+1})\} = pr(V_{n+1} > k).
\]

\noindent Because the original pairs ($\vect{X}_i, Y_i$) are exchangeable, so are the scores $V_i$ for $i \in \mathcal{I}_{cal}$ and $i = n + 1$. So, by~\cref{lemma:inflation_quantiles_ours}, 
\[
pr\{Y_{n+1} \in \hat{C}(\vect{X}_{n+1})\} \geq 1 - \alpha,
\]
\noindent and, under the assumption that the $V_i$'s are almost surely distinct,
\[
pr\{Y_{n+1} \in \hat{C}(\vect{X}_{n+1})\} \leq 1 - \alpha + (n_{cal} + 1)^{-1}.
\]
\end{proof}

\subsection{Proof of Theorem 2}
\begin{proof}
Recall $0<\alpha<1$ is a fixed number. 
    We first claim that:
    \newline {\bf Claim 1}:
    it holds in probability that $\hat{c}(\vect{x})$ converges to $c(\vect{x})$, uniformly in $\vect{x}$ at the rate of $b_{n_{train}}$, as $n_{train}\to \infty$.
    \newline The  proof of the claim will be given below.  
Assuming the validity of the preceding claim and upon noting (C3),  there exists a constant $K$, such that    $V_{i} = \hat{f}(Y_i\mid \vect{X}_i) - \hat{c}(\vect{X}_i)= \tilde{V}_i +\Delta_i$ where $\tilde{V}_i=f(Y_i\mid \vect{X}_i) -c(\vect{X}_i)$ and $\mid \Delta_i\mid \le K \times b_{n_{train}}$, with probability approaching 1, as $n_{train}\to\infty$. 
Consider the empirical distribution $P_{V,n_{cal}}=\sum_{i,=1}^{n_{cal}} n_{cal}^{-1} \delta_{V_i}$, where $\delta$ denotes the Dirac delta probability measure. Its Wasserstein distance from the empirical distribution $P_{\tilde{V},n_{cal}}=\sum_{i=1}^{n_{cal}} n_{cal}^{-1} \delta_{\tilde{V}_i}$ is less than $K\times b_{n_{train}}$, with probability approaching 1. Notice that $P_{V,n_{cal}}$ converges weakly to the distribution of $f(Y\mid \vect{X})-c(\vect{X})$. 
Note that the $\alpha$ quantile $P_{\tilde{V},n_{cal}}$ converges to that of $f(Y\mid \vect{X})-c(\vect{X})$, which is zero, at the rate of $O_p(n_{cal}^{-1/2})$ \citep[Lemma 21.2]{van1998asymptotic}, hence  the stated order of the conformal adjustment $\hat{q}$. 

It remains to verify Claim 1. First note that $c(\vect{x})$ is a continuous function in $\vect{x}$. To see this, let $\vect{x}_m \to \vect{x}$ as $m\to\infty$ and let $\tau_m$ be a sequence of positive numbers that approaches 0 as $m\to \infty$. Furthermore, write $c_m$ for  $c(\vect{x}_m)$. Let $F(\cdot\mid \cdot)$ be the true conditional cdf, and $\hat{F}(\cdot\mid \cdot)$ be that based on $\hat{f}(\cdot\mid \cdot)$. By the definition of $c(\vect{x})$, we have 
\[
\alpha-\tau_m < F(c_m\mid \vect{x}_m
)\le \alpha.\]
Thanks to assumptions ~\ref{assumption: C1} and \ref{assumption: C2}, $f(y\mid \vect{x})$ is a continuous and positive function. Consequently, for all $\tau>0$, $F(c(\vect{x})-\tau\mid \vect{x})<\alpha <F(c(\vect{x})+\tau\mid \vect{x})$. 
It follows from assumptions ~\ref{assumption: C1} and \ref{assumption: C2} and a result of Scheffe\'e \citep[Corollary 2.30]{van1998asymptotic} that $F(c_m\mid \vect{x}_m)-F(c_m\mid \vect{x})\to 0$, as $m\to \infty$. We can then conclude that $\limsup_m F(c_m\mid \vect{x})\le\limsup_m F(c_m\mid \vect{x}_m) + \lim_m \{F(c_m\mid \vect{x})- F(c_m\mid \vect{x}_m) \}\le \alpha$, therefore $c_m<c(\vect{x})+\tau$ eventually. Similarly, we can show that $\liminf F(c_m\mid \vect{x}) \ge \alpha$, hence $c_m>c(\vect{x}-\tau)$ eventually. Since $\tau>0$ is arbitray, we conclude that $c_m\to c(\vect{x})$ as $m\to\infty$. This completes the proof that $c(\vect{x})$ is a continuous function of $\vect{x}$.

Next, we proceed to prove that $\hat{c}(\vect{x})$ converges to $c(\vect{x})$ uniformly in $\vect{x}\in \mathcal{X}$, at the rate of $b_{n_{train}}$, with probability approaching 1. Since $\alpha$ is strictly between 0 and 1, $c(\vect{x})$ lies in $(a,b)$, hence there exists $\psi>0$ such that $\forall \vect{x}\in \mathcal{X}$, $[c(\vect{x})-\psi, c(\vect{x})+\psi]\subset (a,b)$. Since $f(y\mid \vect{x})$ is a continuous, positive function, there exists a positive constant $K_1$ such that 
$f(\cdot\mid \vect{x})$ is bounded below by  $K_1$ over the interval $[c(\vect{x})-\psi, c(\vect{x})+\psi]$. Thus, for all $0\le \eta \le \psi$, we have $F(c(\vect{x})-\eta)\le \alpha -K_1 \eta$ and $F(c(\vect{x})+\eta)\ge \alpha +K_1 \eta$. 
Owing to assumption~\ref{assumption: C3},  $\hat{F}(\cdot\mid \vect{x})$ converges to $F(\cdot\mid \vect{x})$ in sup norm, at the rate of $b_{n_{train}}$ uniformly in $\vect{x}\in \mathcal{X}$, with probability approaching 1. Thus, there exists another positive constant $K_2$ such that for all $0\le \eta \le \psi$, we have $\hat{F}(c(\vect{x})-\eta)\le \alpha -K_1 \eta+ K_2 b_{n_{train}} $ and $\hat{F}(c(\vect{x})+\eta)\ge \alpha +K_1 \eta-K_2 b_{n_{train}} $, hence, for $\eta=2 K_2 K_1^{-1} b_{n_{train}}$, $\hat{F}(c(\vect{x})-\eta)<\alpha$ and $\hat{F}(c(\vect{x})+\eta)>\alpha$, on an event with probability approaching 1. This shows that with probability approaching 1, $\mid \hat{c}(\vect{x})-c(\vect{x})\mid  < 2 K_2 K_1^{-1} b_{n_{train}}$. This completes the proof.
\end{proof}
\subsection{Statement and Proof of Theorem 3}\label{sec: proofs_thm3}
The preceding result can be generalized to the case of parametric estimation of the conditional pdf over a possibly non-compact sample space. Assume the true conditional pdf  belongs  to the model:  $\{f(y\mid\vect{x}; \theta), \theta \in \Theta\}$, with the parameter estimate from the training data denoted as $\hat{\theta}=\hat{\theta}_{train}$. Write the true cutoff as $c(\vect{x};\theta)$, if $\theta$ were the true parameter. Define $H(\eta, \vect{x}, \theta)=\int_{-\infty}^\eta f(y\mid\vect{x}, \theta)dy-\alpha$. Then, $c(\vect{x};\theta)$ is the unique solution of  $H(\cdot, \vect{x},\theta)=0$. The smoothness property of $c(\cdot;\cdot)$ can be deduced via the implicit function theorem.  The following regularity conditions are mild conditions, generally satisfied with maximum likelihood estimation of $\theta$. Below, denote the true parameter as $\theta_0$.
\setcounter{assumption}{4}
\begin{assumption}~\label{assumption: C5}
    The estimate $\hat{\theta}\to \theta_0$ almost surely. Furthermore, 
   $f(y\mid\vect{x}; \hat{\theta})\to f(y\mid\vect{x};\theta_0)$ uniformly for $y, \vect{x}$ in any fixed compact set, as $n_{train}\to\infty$.
\end{assumption}
\begin{assumption}~\label{assumption: C6}
 The conditional pdf $f(y\mid\vect{x}, \theta)$ is a positive, continuous function of $y, \vect{x}, \theta$. The functions $\partial H/\partial x= \int_{-\infty}^\eta \partial f(y\mid\vect{x};\theta)/\partial \vect{x}\,\,dy$ and 
$\partial H/\partial \theta= \int_{-\infty}^\eta \partial f(y\mid\vect{x};\theta)/\partial \theta \,\,dy$ are well-defined and continuous functions of  $\eta, \vect{x}, \theta$.    
\end{assumption}

\begin{theorem}\label{thm: convergence}
Suppose 
assumptions~\ref{assumption: C4}--\ref{assumption: C6}
 hold, and let $1>\alpha>0$ be fixed. Then $\hat{c}(\vect{x})=c(\vect{x}, \hat{\theta})$ converges to $c(\vect{x}, \theta_0)$  uniformly for $\vect{x}$ in any compact set.  Moreover, the conformal adjustment $\hat{q}\to 0$, as both $n_{cal}\to\infty$ and $n_{train}\to\infty$. 
\end{theorem}

\begin{proof}
    We first verify that $c(\vect{x}; \theta)$ is a continuous function. Recall $c(\vect{x}; \theta)$ is the unique solution to the equation $H(\eta, \vect{x}, \theta)=0$. Now, ${\partial H}/{\partial \eta} =f(\eta\mid \vect{x}; \theta)>0$, for any $\eta$. It follows from the implicit function theorem and assumption~\ref{assumption: C6} that   $c(\vect{x}; \theta)$ is a  continuously differentiable function of $\vect{x}, \theta$. Hence, $c(\vect{x}; \hat{\theta})\to c(x, \theta_0)$ uniformly for $\vect{x}$ in any fixed compact set, thanks to assumption~\ref{assumption: C5}.

Consider the empirical distribution $P_{V,n_{cal}}=\sum_{i,=1}^{n_{cal}} n_{cal}^{-1} \delta_{V_i}$, where $V_{i} = f(Y_i\mid \vect{X}_i; \hat{\theta}) - c(\vect{X}_i; \hat{\theta})$ and $\delta$ denotes the Dirac delta probability measure. 
Next, we show that  as $n_{train}\to \infty$ and $n_{cal}\to \infty$,  $P_{V,n_{cal}}$ converges weakly to the distribution of $f(Y\mid \vect{X};\theta_0) - c(\vect{X}; \theta_0)$. It suffices to show that for any bounded continuous function $g$, 
\begin{equation}
 \sum_{i,=1}^{n_{cal}} g\{f(Y_i\mid \vect{X}_i,\hat{\theta}) - c(\vect{X}_i; \hat{\theta})\}/n \to E\{g(f(Y\mid \vect{X}; \theta_0) -c(\vect{X}; \theta_0)\}.
    \label{eq: weak convergence}
\end{equation}
Let $g$ be upper bounded by a constant $M>0$.
Let $\epsilon>0$ be given and  $\mathcal{R}$ be a compact set such that with probability greater than $1-\epsilon/M$, $(Y, \vect{X}^\intercal)^\intercal\in \mathcal{R}$. Decompose
$\sum_{i,=1}^{n_{cal}} g\{\hat{f}(Y_i\mid \vect{X}_i) - c(\vect{X}_i; \hat{\theta})\}/n =S_1+S_2+S_3$ where 
\begin{eqnarray*}
    S_1 &=& \sum_{i,=1}^{n_{cal}} [g\{f(Y_i\mid \vect{X}_i; \hat{\theta}) - c(\vect{X}_i; \hat{\theta})\}-
    g\{f(Y_i\mid \vect{X}_i; \theta_0) - c(\vect{X}_i; \theta_0)\}]
    I\{(Y_i, \vect{X}_i^\intercal)^\intercal\in \mathcal{R}\} /n, \\
    S_2 &=& \sum_{i,=1}^{n_{cal}} [g\{f(Y_i\mid \vect{X}_i; \hat{\theta}) - c(\vect{X}_i; \hat{\theta})\}-
    g\{f(Y_i\mid \vect{X}_i; \theta_0) - c(\vect{X}_i; \theta_0)\}]
    I\{(Y_i, \vect{X}_i^\intercal)^\intercal\not \in \mathcal{R}\} /n, \\
    S_3&=& \sum_{i,=1}^{n_{cal}} g\{f(Y_i\mid \vect{X}_i; \theta_0) - c(\vect{X}_i; \theta_0)\}/n.
\end{eqnarray*}
 Because of assumption~\ref{assumption: C5}, $T_1\to 0$, almost surely, as $n_{train}\to \infty$.
 Note that $T_2\to E(T_2)$ and $E(\mid T_2\mid )\le \epsilon$, while $T_3\to E[\{g(f(Y\mid \vect{X}; \theta_0) -c(\vect{X}; \theta_0)\}$, as  $n_{cal}\to\infty$.   Since $\epsilon>0$, \eqref{eq: weak convergence} holds. 

 The preceding weak convergence result and assumption~\ref{assumption: C4} then imply that the $\alpha$ quantile of $P_{\tilde{V},n_{cal}}$ converges to that of $f(Y\mid \vect{X}; \theta_0)-c(\vect{X}; \theta_0)$, which is zero \citep[Lemma 21.2]{van1998asymptotic}, hence  the conformal adjustment $\hat{q}\to 0$. 
\end{proof}

\subsection{Comparison of Theoretical Results}

The result of Theorem 2 and Theorem 3 are similar to Theorem 25 in \cite{izbicki2021cdsplit} and Theorem 2 part 1 in \cite{CHR}, though we have a slightly stronger assumption on the conditional estimator (density and quantile, respectively). Theorem 1 in \cite{CQR_theory} and Theorem 5 in \cite{dcp} are slightly stronger results, that the Lebeguse measure of the set difference vanishes asymptotically, though it only applies to general quantiles for CQR and not the shortest interval. It does apply to the oracle shortest interval for optimal DCP. 

\section{Numerical Studies}\label{sec:further-chcds_simulation}

An R package that implements CHCDS can be found on \href{http://www.github.com/maxsampson/CHCDS}{GitHub here}. In this section, we include a second simulation scenario to go along with the mixture scenario found in Section~\ref{sec:chcds_simulation}. The covariate was generated in the same way, $X \sim \text{Unif}(-1.5, 1.5)$. The conditional response is given below. A scatterplot of the covariate-response relationship can be found in~\cref{fig:scatterplot_asymm}. A scatterplot of the covariate-response relationship in the mixture scenario can be found in~\cref{fig:scatterplot_mix}.

\noindent Asymmetric: $Y|X = 5 + 2X + \epsilon|X$,  $\epsilon|X \sim \text{Gamma}(\text{Shape} = 1 + 2 |X|, \text{ Rate} = 1 + 2|X|)$ \newline

\begin{figure}[ht]
    \centering
\includegraphics[scale = 0.5]{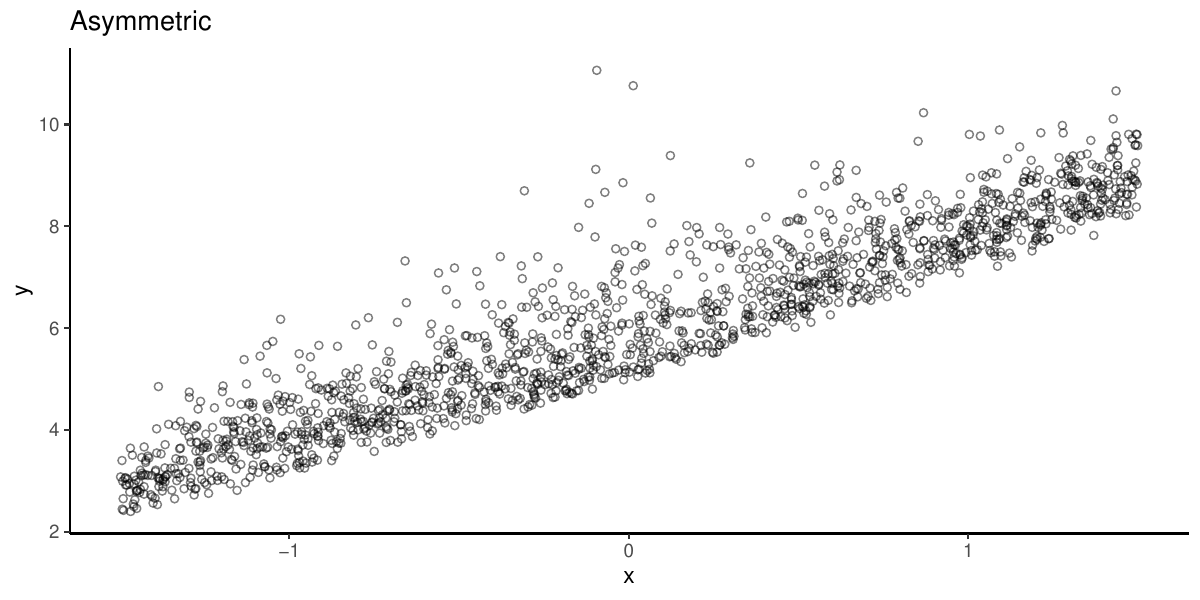}
\caption{A plot showing the covariate response relationship from one simulation in the asymmetric scenario.}\label{fig:scatterplot_asymm}
\end{figure}

\begin{figure}[ht]
    \centering
\includegraphics[scale = 0.5]{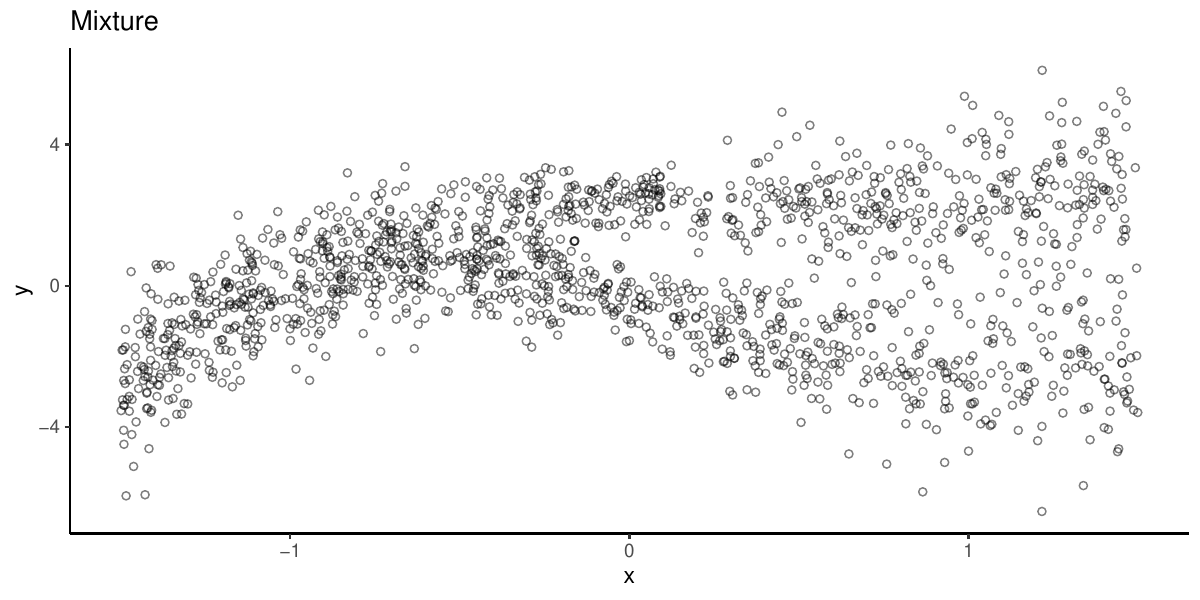}
\caption{A plot showing the covariate response relationship from one simulation in the mixture scenario.}\label{fig:scatterplot_mix}
\end{figure}

Our setup was the same as in the main text, 1,000 points for training the model and 500 for calibrating the non-conformity scores. The models used for HPD-split, DCP, CQR, and CHR were also the same. We also included how well CHCDS performed with two additional conditional density estimators, a conditional kernel density estimator (Kernel) and the FlexCode density estimator (FlexCode). We also included the results of these conditional density estimators with CHCDS in the mixture scenario. Results can be found for the asymmetric scenario in~\cref{tab:CDE_HPD_asymm} and the mixture scenario in~\cref{tab:CDE_HPD_Bimodal_full}. Conditional coverage plots can be found in~\cref{fig:conditional_cov_asymmetric} and~\cref{fig:conditional_cov_mixture_full}

\begin{table}[ht]
\begin{center}
\caption{Comparison of the methods in the asymmetric scenario}

    \begin{tabular}{lccc}
         Approach &  Coverage & Set Size & Conditional Absolute Deviation   \\
        Unadjusted (FlexCode) & 969 (2) & 4054 (2) & 74 (1)\\
        HPD-split (FlexCode) & 908 (3) & 2307 (2) & 30 (1)\\
        CHCDS (FlexCode) & 907 (3) & 2253 (2) & 29 (1) \\
        CHCDS (Kernel) & 900 (3) & 1957 (1) & 14\\
        CHCDS (KNN) & 900 (3) & 1944 (1) & 12\\
        CHCDS (Gaussian Mix) & 902 (3) & 2005 (1) & 3 (1)\\
        DCP & 901 (3) & 2031 (1) & 30 (1) \\
        CQR & 900 (3) & 2246 (1) & 8  \\
        CHR & 898 (3) & 2334 (2) & 11 \\
    \end{tabular}
    \label{tab:CDE_HPD_asymm}
\end{center}
    \caption{
	Monte Carlo error given in parentheses only if it is greater than 1. All values have been
multiplied by $10^3$.
}

\end{table}

\begin{table}[ht]
\begin{center}
\caption{Comparison of the methods in the mixture scenario with additional methods}
    \begin{tabular}{lccc}
         Approach &  Coverage & Set Size & Conditional Absolute Deviation   \\
        Unadjusted (FlexCode) & 925 (3) & 5878 (2) & 58 (1)\\
        HPD-split (FlexCode) & 900 (3) & 5466 (3) & 63 (1)\\
        CHCDS (FlexCode) & 904 (3) & 5526 (3) & 64 (1)  \\
        CHCDS (Kernel) & 904 (3) & 5668 (3)  & 79 (1) \\
        CHCDS (KNN) & 906 (3) & 5319 (2) & 8 \\
        CHCDS (Gaussian Mix) & 903 (3) & 5156 (2) & 28 (1)\\
        DCP & 901 (3) & 5173 (1) & 5  \\
        CQR & 901 (3) & 5834 (2) & 21 (1) \\
        CHR & 899 (3) & 5741 (2) & 14 \\
    \end{tabular}
    \label{tab:CDE_HPD_Bimodal_full}
\end{center}
    \caption{
	Monte Carlo error given in parentheses only if it is greater than 0.001. All values have been
multiplied by $10^3$.
}
\end{table}

 \begin{figure}[ht]
     \centering
     \begin{tabular}{cc}
  \includegraphics[scale = 0.32]{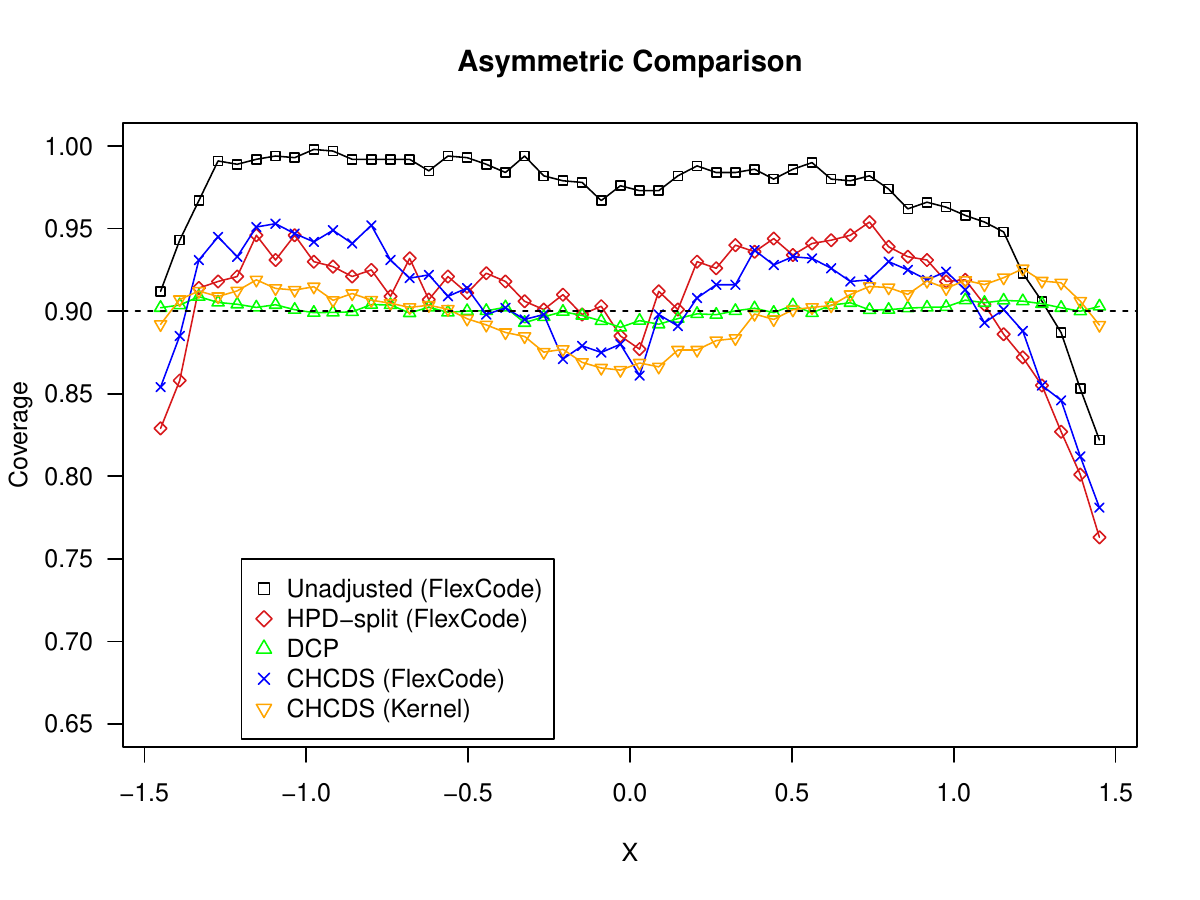}       &  \includegraphics[scale = 0.32]{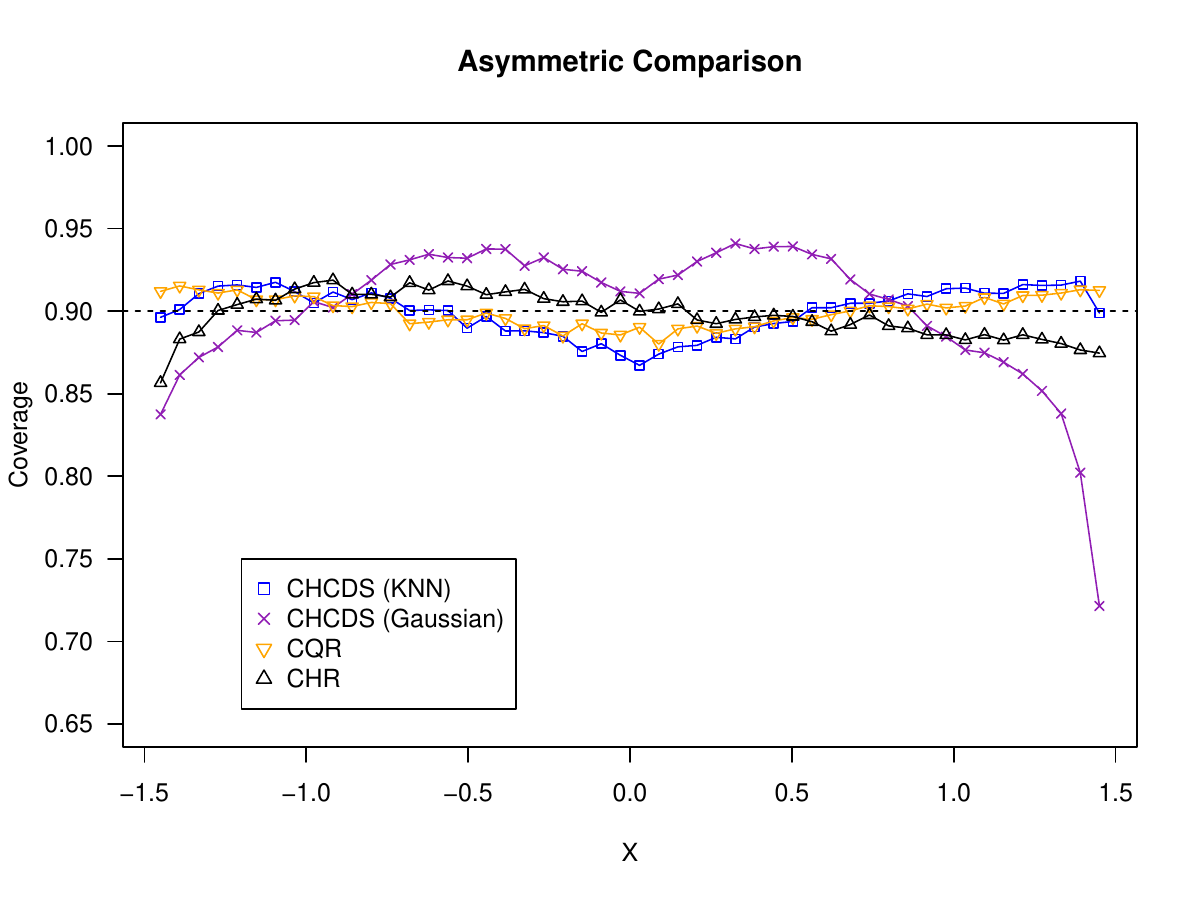}\\
  a &b 
    \end{tabular}
     \caption{A diagram showing the comparison of conditional coverage in the mixture scenario. The dashed line represents the desired 90\% coverage. The other lines represent the conditional coverage at a given value of $X$.}\label{fig:conditional_cov_asymmetric}
\end{figure}

 \begin{figure}[ht]
     \centering
     \begin{tabular}{cc}
  \includegraphics[scale = 0.32]{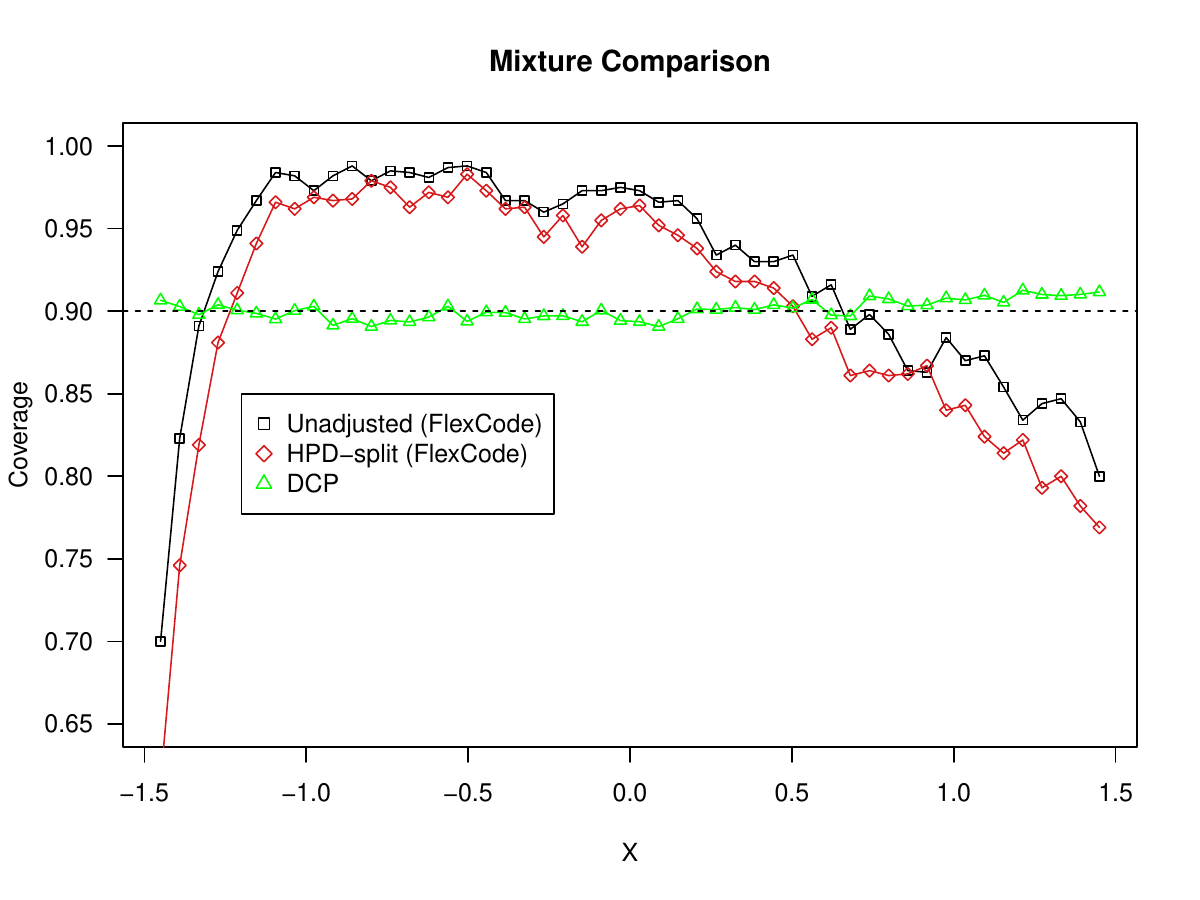}       &  \includegraphics[scale = 0.32]{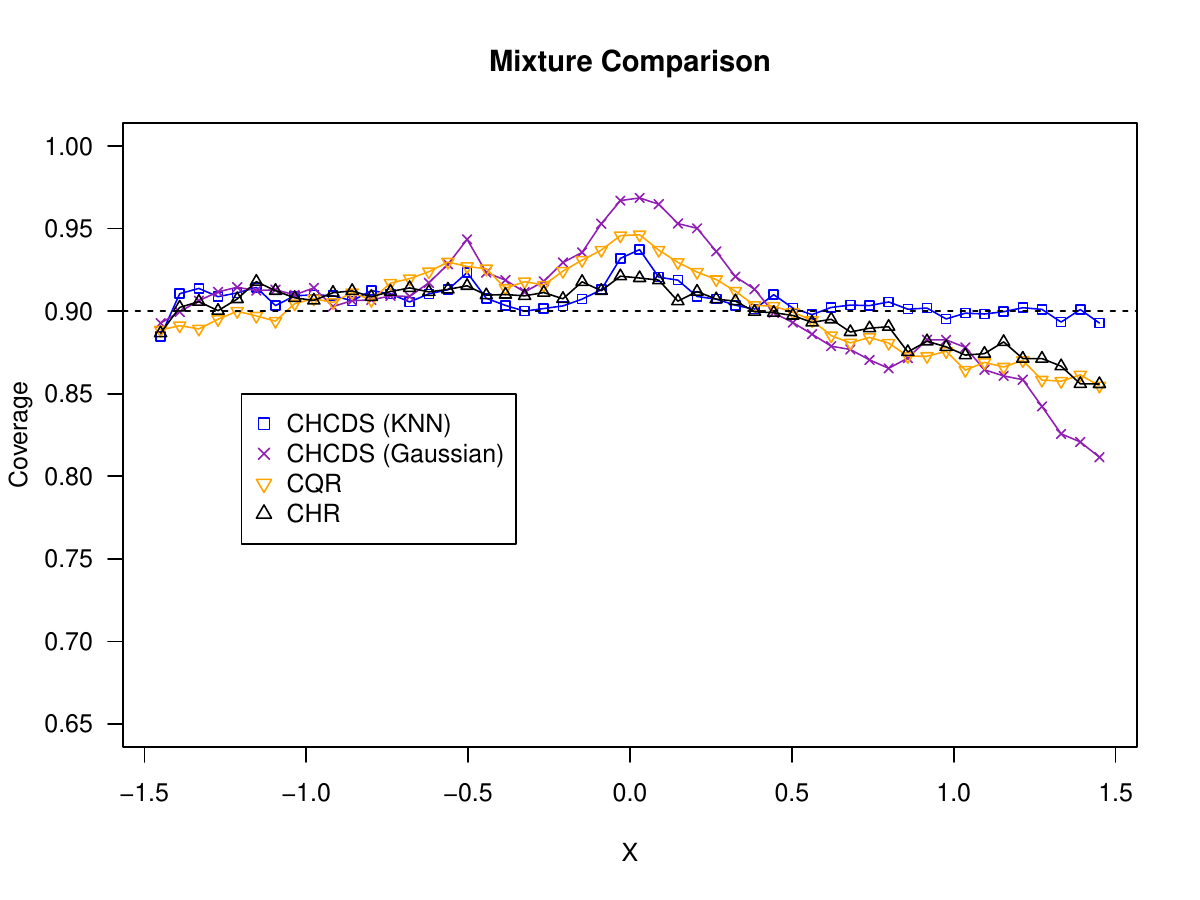}\\
  a &b 
    \end{tabular}
    \caption{A diagram showing the comparison of conditional coverage in the mixture scenario. The dashed line represents the desired 90\% coverage. The other lines represent the conditional coverage at a given value of $X$.}\label{fig:conditional_cov_mixture_full}
\end{figure}

As observed in Section~\ref{sec:chcds_simulation} 
in the main text, nearly all of the methods have some undercoverage in the tails in at least one of the scenarios. 
In the mixture scenario, if we wanted to predict our response at $X = 0.5$, an observed data point with a covariate whose value is $-0.5$ does not provide value because the conditional distributions are completely different. This is why the method that applied the heaviest local weight in the mixture scenario, KNN kernel conditional density estimator with our conformal adjustment, performed very well. 

The conditional coverage seen in~\cref{fig:conditional_cov_asymmetric} and~\cref{fig:conditional_cov_mixture_full}, where the coverage rates fall off in the tails of distributions, is common with conformal prediction \citep{Lei_Wasserman_Conformal_kernel, dis_free_pred_COPS}. It can be thought of as a bit of extrapolation, even though we observe data in the region. The smaller coverage can be seen well in~\cref{fig:regions_HPD_asymm} and~\cref{fig:regions_our_asymm}. Even though we used a weighted regression technique to compute the FlexCode estimators, the density estimators that used larger local weights worked better. This can be seen in~\cref{fig:regions_kernel_asymm} and~\cref{fig:regions_knn_asymm}. Though, the techniques that use local weights do not always perform well in high dimensions \citep{Wang_Scott_2019_nonparametric_density}. Which is an important reminder to choose an appropriate conditional density estimator for the given data.

One other noteworthy point is that when the density estimators are the same, CHCDS and HPD-split provide similar set sizes as well as conditional coverage. The two advantages of CHCDS compared to HPD-split are the easy to understand conformal adjustment to an existing highest density set, and that it does not require numerical integration because estimating the cutoff point for highest density set does not require numerical integration \citep{hyndman_conditional_density_1996}. 

Below, we include several displays of the prediction sets in~\cref{fig:regions_unadj_asymm} - ~\cref{fig:regions_CHR_mixture}

\begin{figure}[ht]
    \centering
\includegraphics[scale = 0.5]{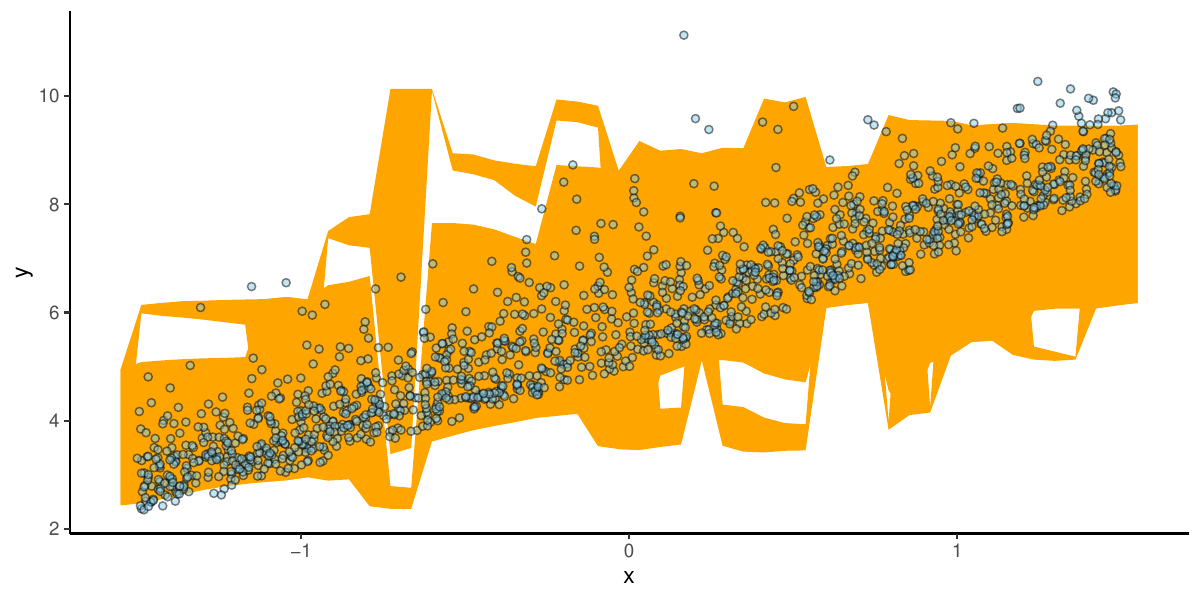}
\caption{An example of the prediction regions given by Unadjusted (FlexCode) in the asymmetric scenario.}\label{fig:regions_unadj_asymm}
\end{figure}
\begin{figure}[ht]
    \centering
\includegraphics[scale = 0.5]{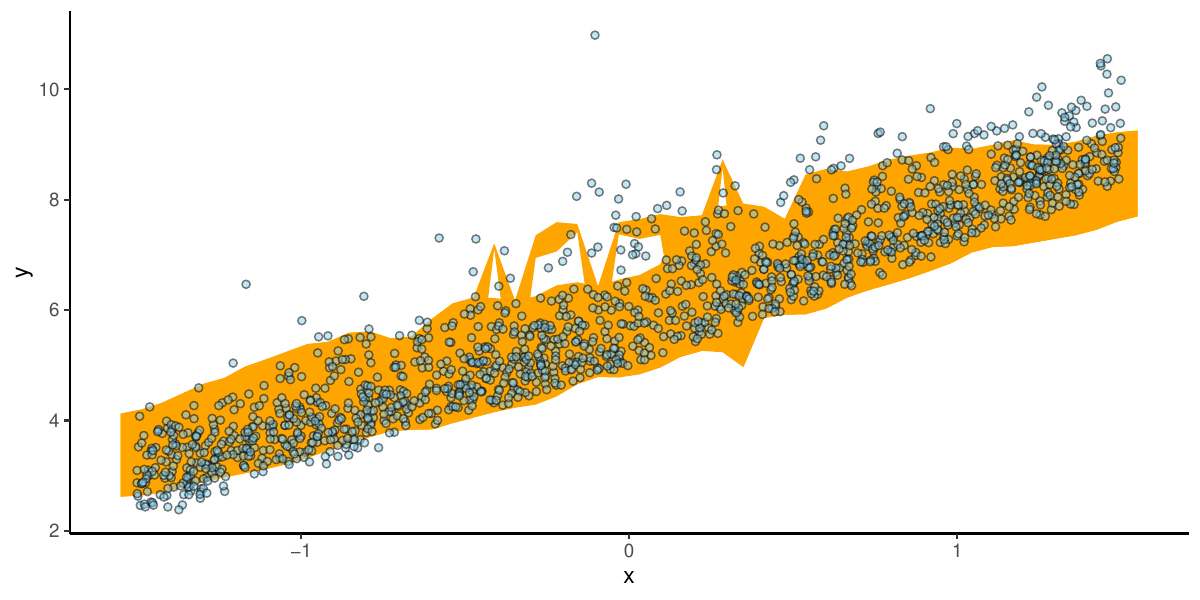}
\caption{An example of the prediction regions given by HPD-split (FlexCode) in the asymmetric scenario.}\label{fig:regions_HPD_asymm}
\end{figure}

\begin{figure}[ht]
    \centering
\includegraphics[scale = 0.5]{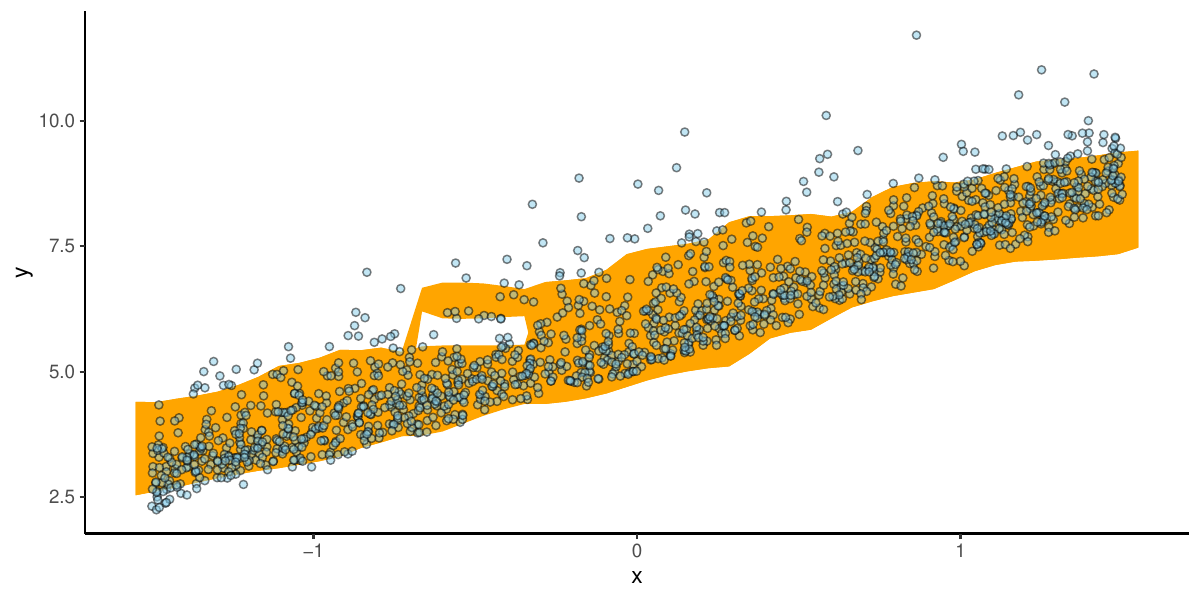}
\caption{An example of the prediction regions given by CHCDS (FlexCode) in the asymmetric scenario.}\label{fig:regions_our_asymm}
\end{figure}

\begin{figure}[ht]
    \centering
\includegraphics[scale = 0.5]{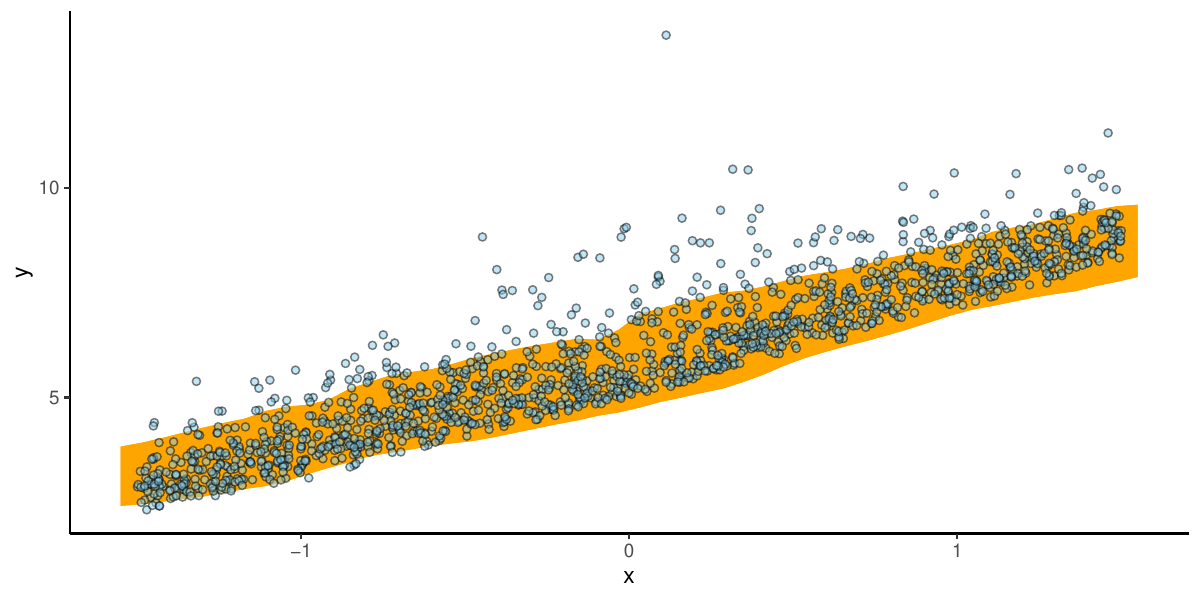}
\caption{An example of the prediction regions given by CHCDS (Kernel) in the asymmetric scenario.}\label{fig:regions_kernel_asymm}
\end{figure}

\begin{figure}[ht]
    \centering
\includegraphics[scale = 0.5]{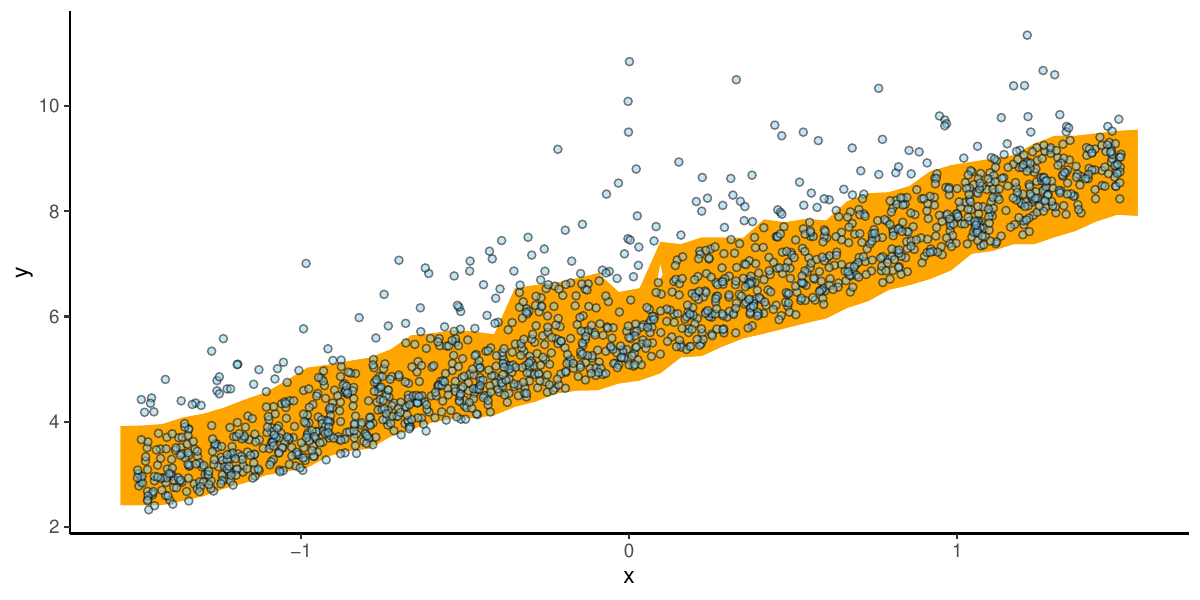}
\caption{An example of the prediction regions given by CHCDS (KNN) in the asymmetric scenario.}\label{fig:regions_knn_asymm}
\end{figure}

\begin{figure}[ht]
    \centering
\includegraphics[scale = 0.5]{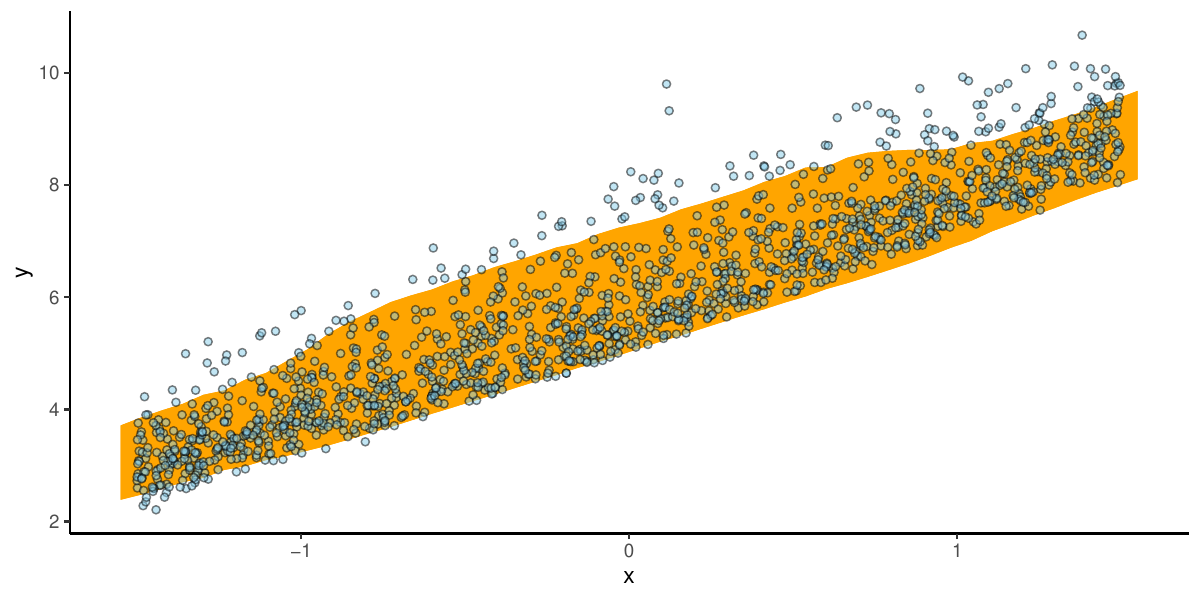}
\caption{An example of the prediction regions given by CHCDS (Gaussian Mix) in the asymmetric scenario.}\label{fig:regions_gauss_asymm}
\end{figure}

\begin{figure}[ht]
    \centering
\includegraphics[scale = 0.5]{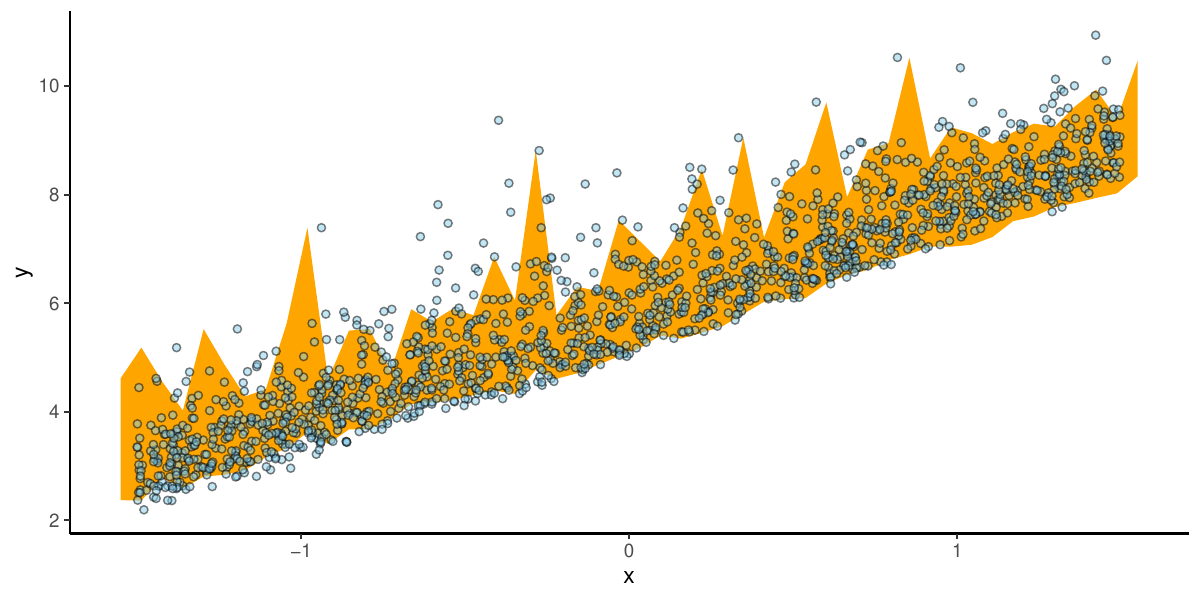}
\caption{An example of the prediction regions given by DCP in the asymmetric scenario.}\label{fig:regions_DCP_asymm}
\end{figure}

\begin{figure}[ht]
    \centering
\includegraphics[scale = 0.5]{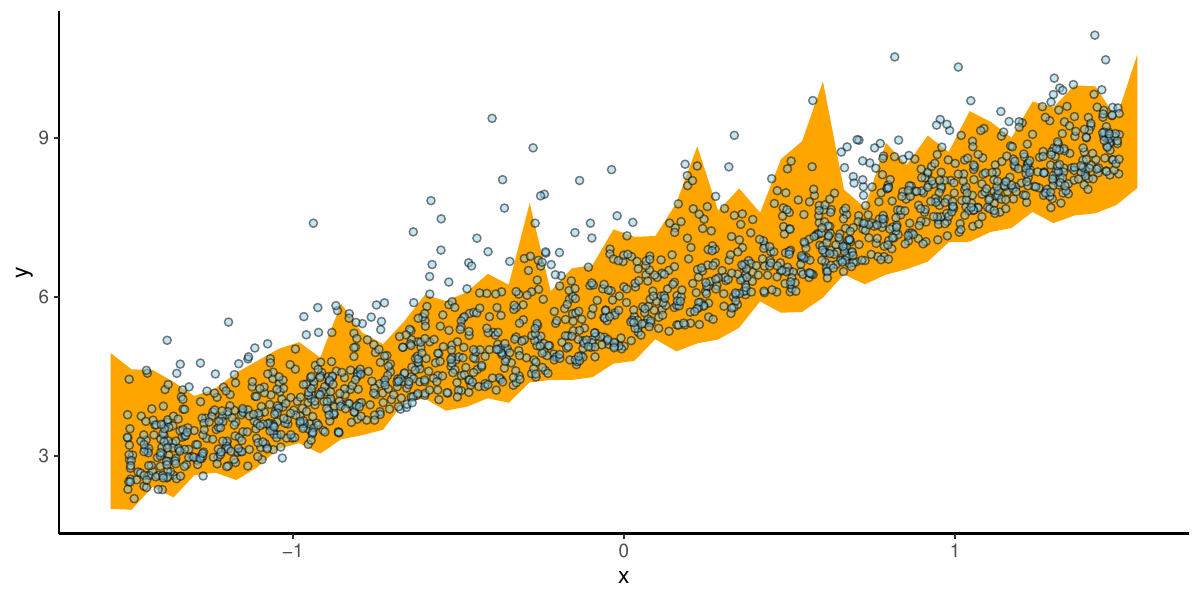}
\caption{An example of the prediction regions given by CQR in the asymmetric scenario.}\label{fig:regions_CQR_asymm}
\end{figure}

\begin{figure}[ht]
    \centering
\includegraphics[scale = 0.5]{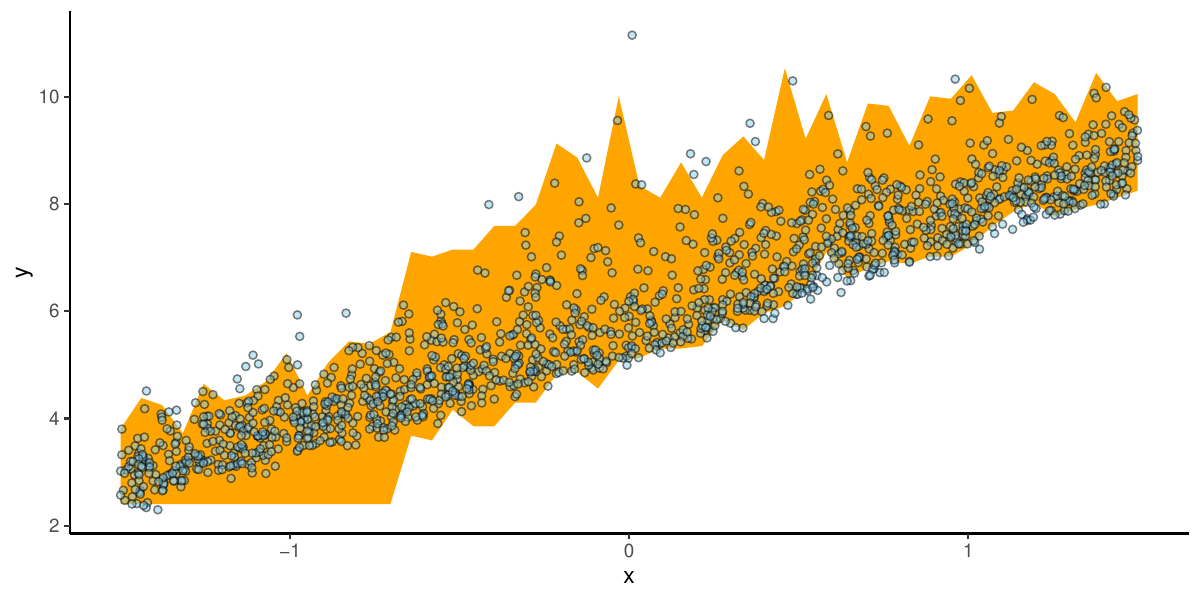}
\caption{An example of the prediction regions given by CHR in the asymmetric scenario.}\label{fig:regions_CHR_asymm}
\end{figure}

\begin{figure}[ht]
    \centering
\includegraphics[scale = 0.5]{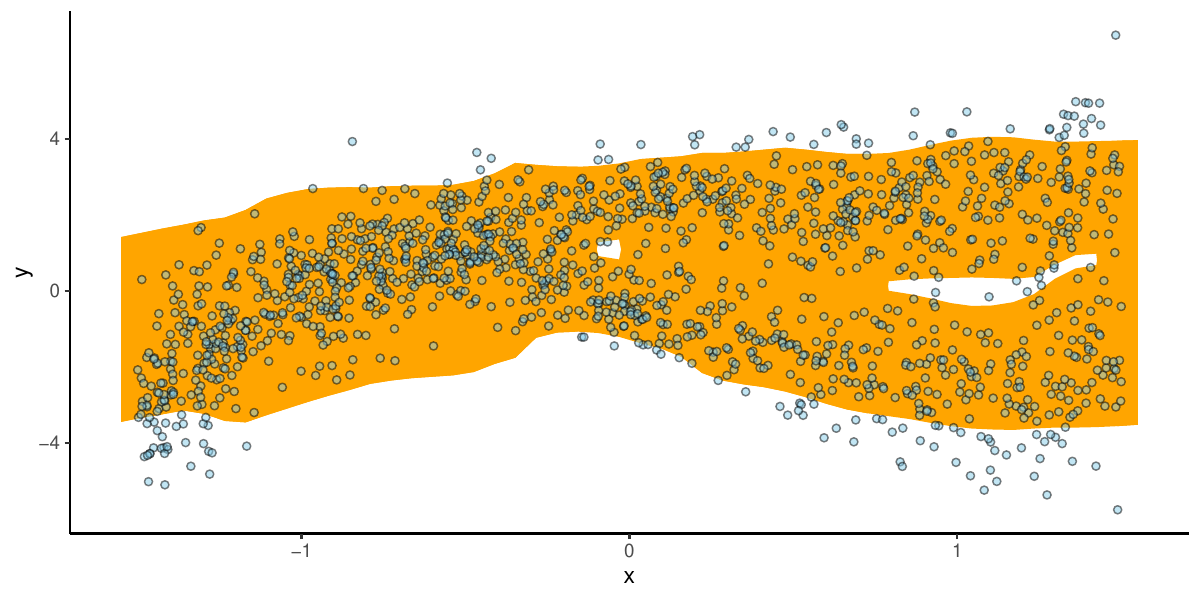}
\caption{An example of the prediction regions given by Unadjusted (FlexCode) in the mixture scenario.}\label{fig:regions_unadj_mixture}
\end{figure}

\begin{figure}[ht]
    \centering
\includegraphics[scale = 0.5]{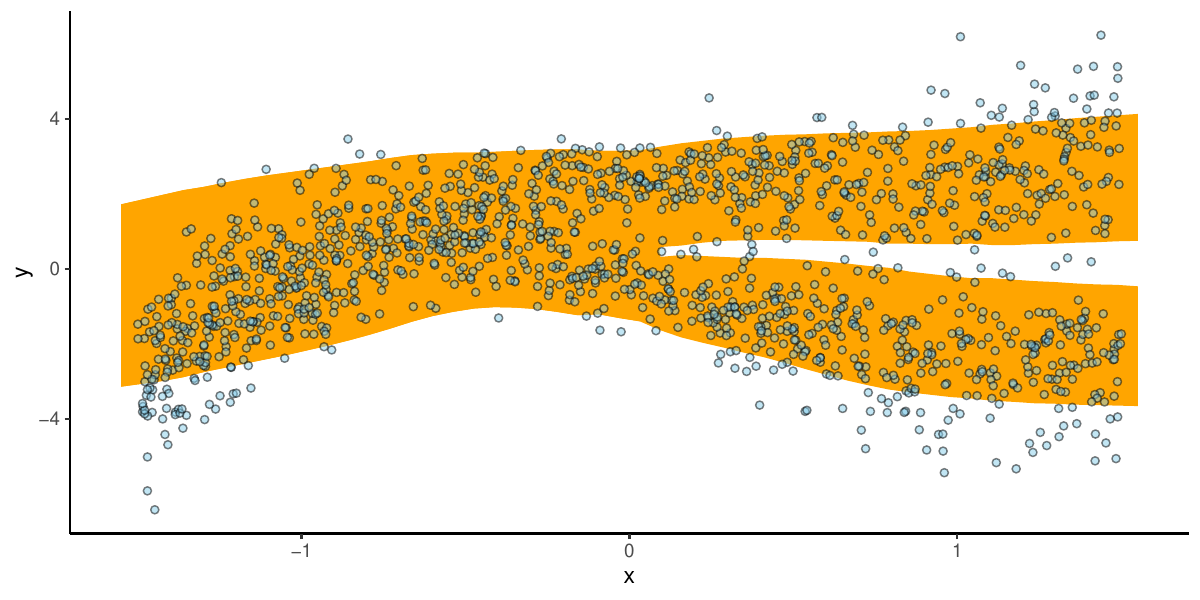}
\caption{An example of the prediction regions given by HPD-split (FlexCode) in the mixture scenario.}\label{fig:regions_HPD_mixture}
\end{figure}

\begin{figure}[ht]
    \centering
\includegraphics[scale = 0.5]{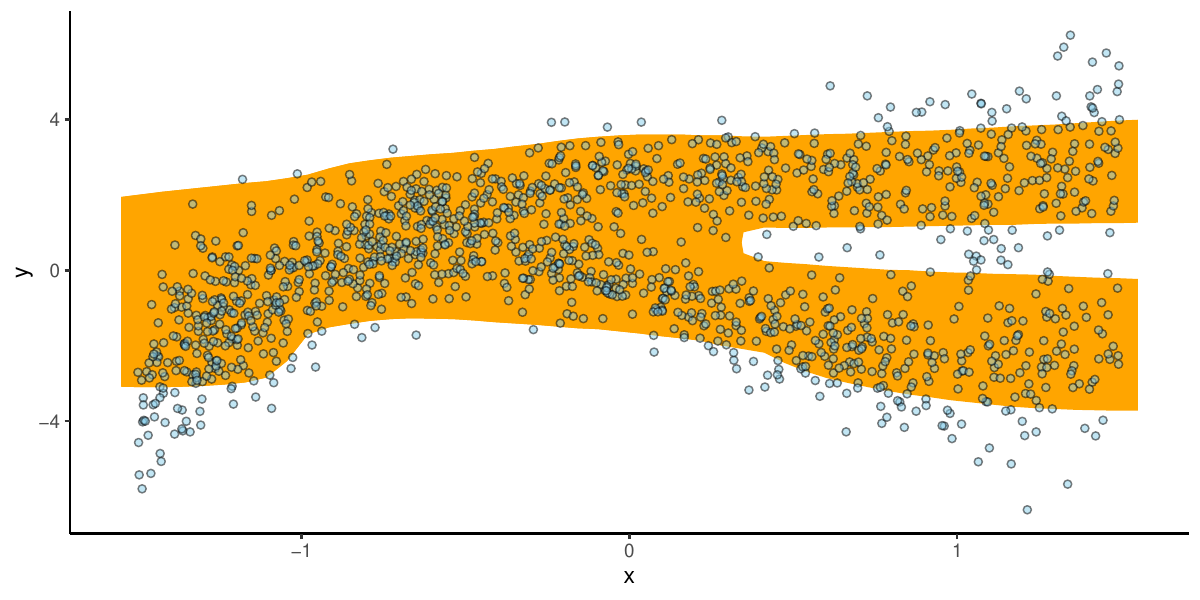}
\caption{An example of the prediction regions given by CHCDS (FlexCode) in the mixture scenario.}\label{fig:regions_our_mixture}
\end{figure}

\begin{figure}[ht]
    \centering
\includegraphics[scale = 0.5]{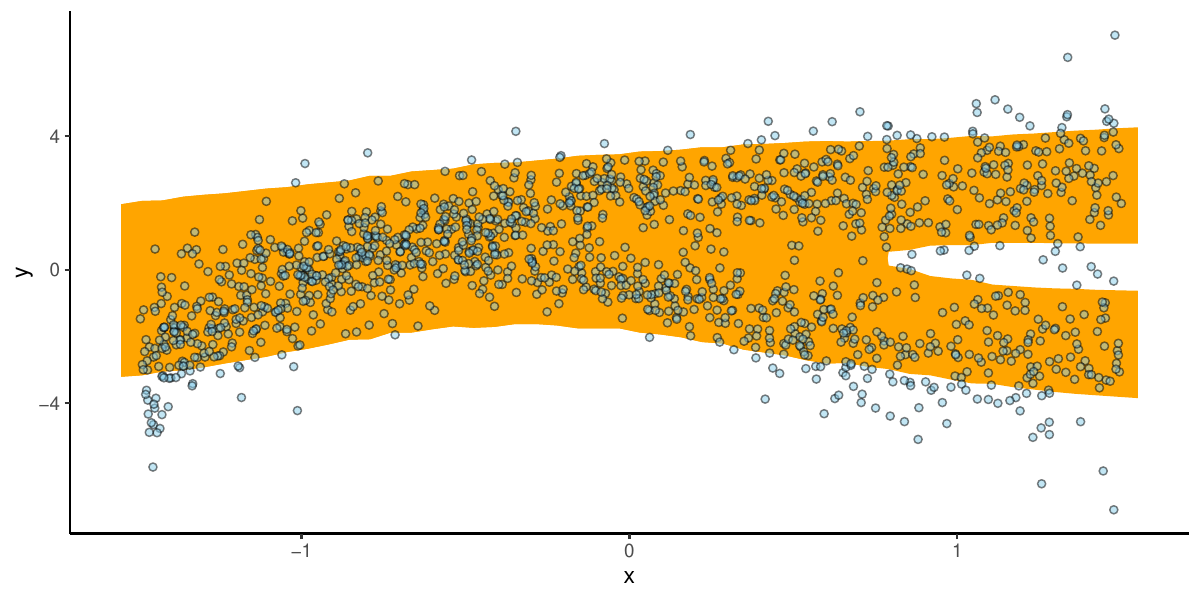}
\caption{An example of the prediction regions given by CHCDS (Kernel) in the mixture scenario.}\label{fig:regions_kernel_mixture}
\end{figure}

\begin{figure}[ht]
    \centering
\includegraphics[scale = 0.5]{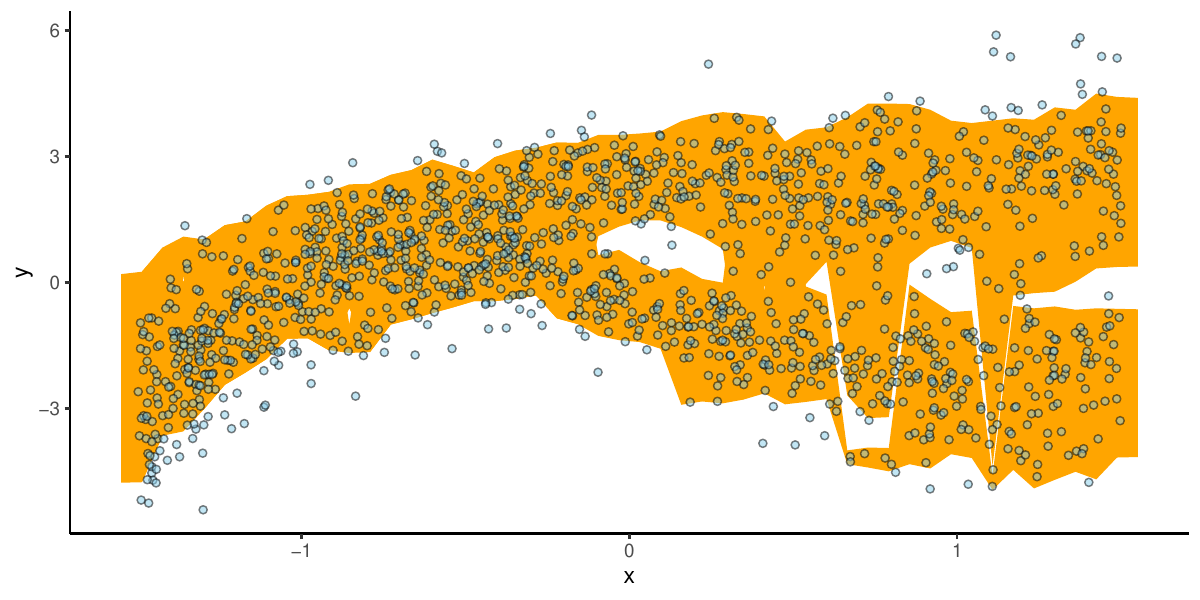}
\caption{An example of the prediction regions given by CHCDS (KNN) in the mixture scenario.}\label{fig:regions_knn_mixture}
\end{figure}

\begin{figure}[ht]
    \centering
\includegraphics[scale = 0.5]{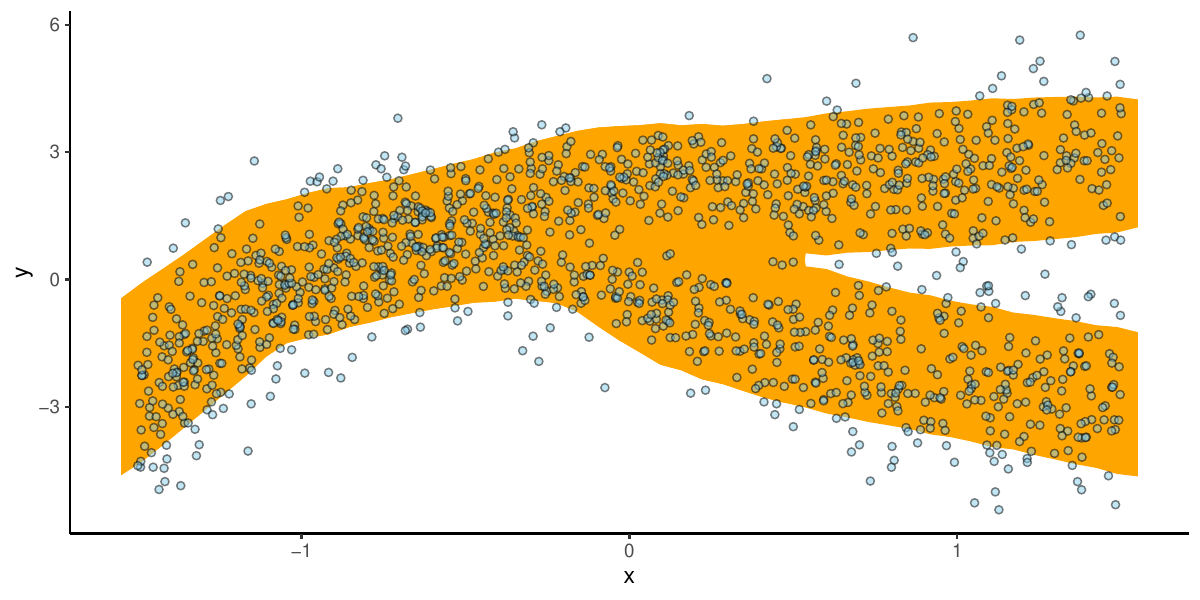}
\caption{An example of the prediction regions given by CHCDS (Gaussian Mix) in the mixture scenario.}\label{fig:regions_gaus_mixture}
\end{figure}

\begin{figure}[ht]
    \centering
\includegraphics[scale = 0.5]{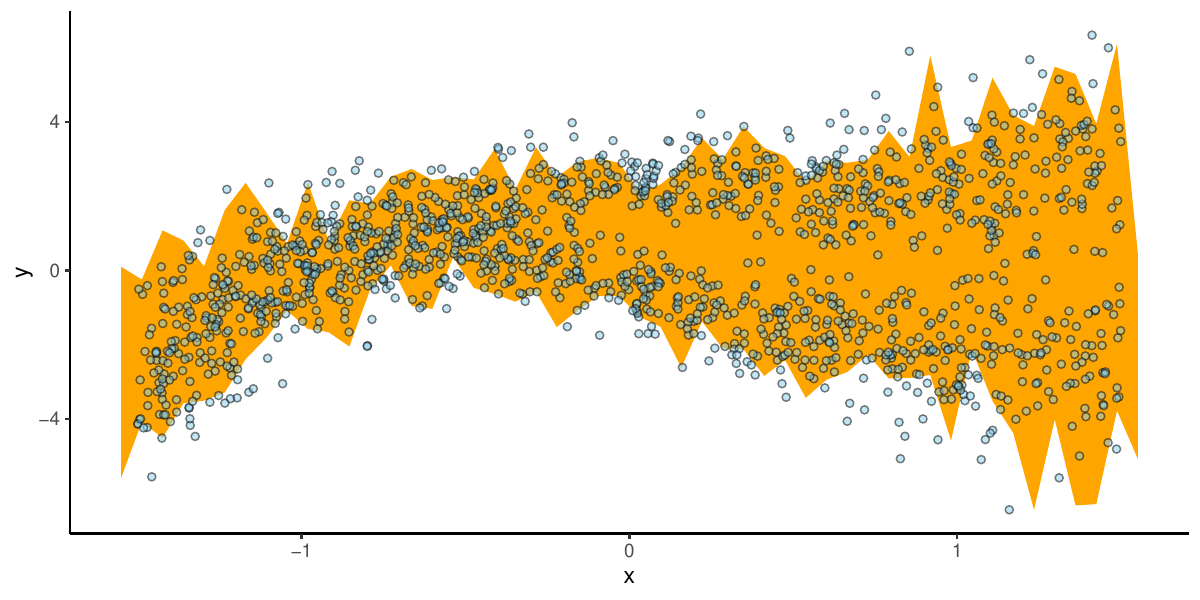}
\caption{An example of the prediction regions given by DCP in the mixture scenario.}\label{fig:regions_DCP_mixture}
\end{figure}

\begin{figure}[ht]
    \centering
\includegraphics[scale = 0.5]{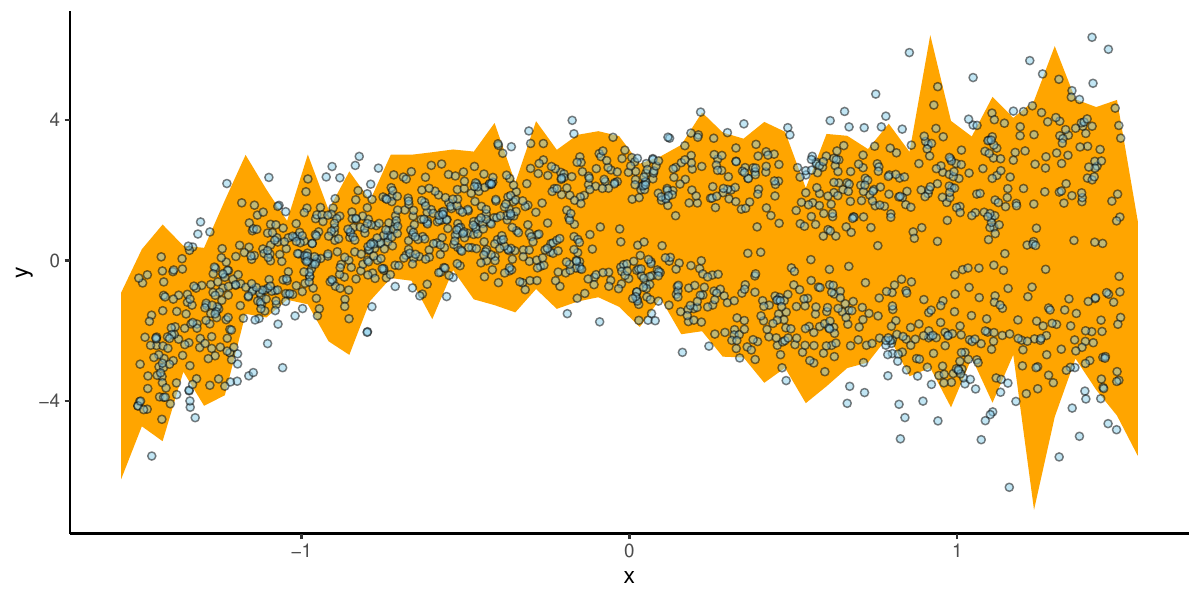}
\caption{An example of the prediction regions given by CQR in the mixture scenario.}\label{fig:regions_CQR_mixture}
\end{figure}

\begin{figure}[ht]
    \centering
\includegraphics[scale = 0.5]{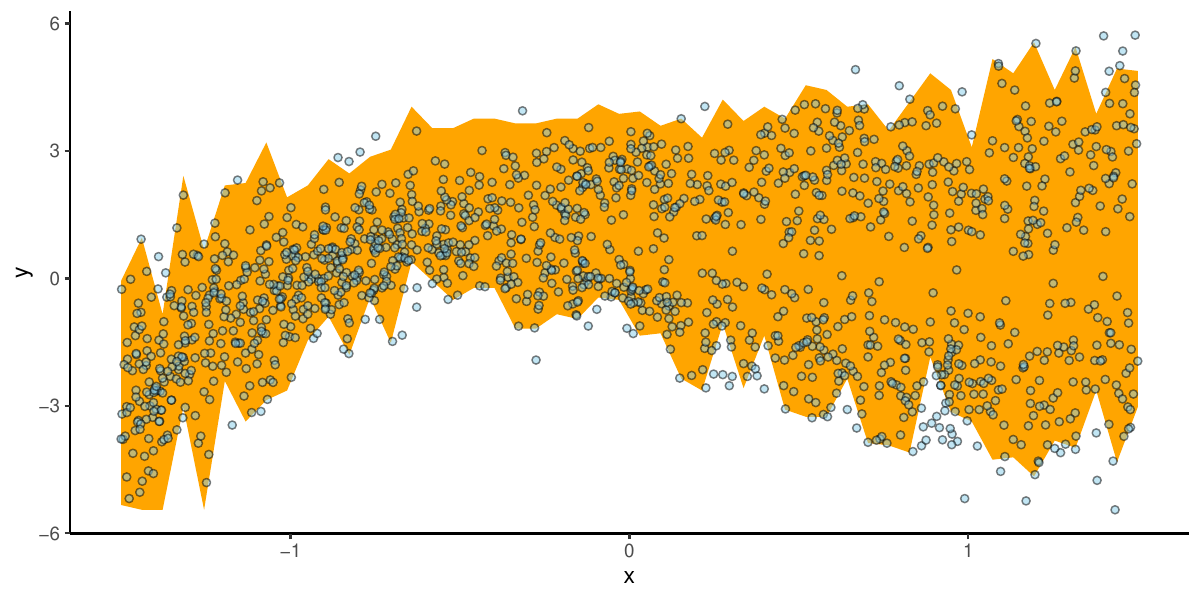}
\caption{An example of the prediction regions given by CHR in the mixture scenario.}\label{fig:regions_CHR_mixture}
\end{figure}

\section{Further Real Analysis}~\label{sec:further data analysis}
To demonstrate the flexibility of CHCDS, we also created a conditional Gaussian mixture density estimator in Pytorch \citep{pytorch_citation}. We trained the model on 64,950 observations, computed calibration scores on 5,000 observations, and computed the coverage results on 5,000 out of sample observations. Those results can be found in~\cref{tab:redshift_pytorch}. The feedforward model had 4 components, 3 hidden layers, and  2,652 parameters. The code can be found in a Jupyter notebook on  \href{https://github.com/maxsampson/CHCDS_HappyA}{GitHub here}. 

\begin{table}[ht]
\begin{center}
\caption{Coverage, conditional coverage, and average size of the prediction regions for CHCDS Pytorch conditional Gaussian mixture density estimator.}

    \begin{tabular}{cc}
         Coverage & Size \\ 
         802 (6) & 93 (1)  \\
         Coverage (Bright) & Size (Bright) \\
         800 (8) & 57 \\
         Coverage (Faint) & Size (Faint)   \\
         804 (8) & 129 (2) \\
    \end{tabular}
    \label{tab:redshift_pytorch}
\end{center}
    \caption{
	Standard errors are given in parentheses if they are greater than 1. All values have been multiplied by $10^3$.
	}
\end{table}

\newpage
\bibliography{ref}
\end{document}